%% file: neurips_2026.tex
\documentclass{article}

 \usepackage[preprint]{neurips_2026}

\usepackage[utf8]{inputenc} 
\usepackage[T1]{fontenc}    
\usepackage{hyperref}       
\usepackage{url}            
\usepackage{booktabs}       
\usepackage{amsfonts}       
\usepackage{nicefrac}       
\usepackage{microtype}      
\usepackage{xcolor}         

\usepackage{multirow}
\usepackage{amsmath}
\usepackage{amssymb}
\usepackage{amsthm}         
\usepackage{graphicx}       
\usepackage{subcaption}     
\usepackage{adjustbox}      
\usepackage{makecell}       
\usepackage{rotating}       
\usepackage{siunitx}
\usepackage{wrapfig}

\newtheorem{proposition}{Proposition}
\newtheorem{corollary}{Corollary}

\providecommand{\Description}[1]{}
\providecommand{\ccsdesc}[2][]{}
\providecommand{\keywords}[1]{}
\providecommand{\acmConference}[4][]{}
\providecommand{\acmDOI}[1]{}
\providecommand{\acmISBN}[1]{}
\providecommand{\acmYear}[1]{}
\providecommand{\copyrightyear}[1]{}
\providecommand{\setcopyright}[1]{}
\providecommand{\citestyle}[1]{}
\providecommand{\received}[2][]{}
\providecommand{\acmBooktitle}[1]{}
\providecommand{\acmPrice}[1]{}

\title{What Gets Measured Gets Managed: Sign-aware Recommendation Needs Sign-aware Evaluation}

\author{%
  Minchan Kim \\
  Graduate School of Data Science\\
  Seoul National University\\
  Seoul, Republic of Korea \\
  \texttt{mmm5373@snu.ac.kr} \\
  \And
  Jungmin Hwang \\
  Department of Data Science\\
  Seoul National University of Science and Technology\\
  Seoul, Republic of Korea \\
  \texttt{inextro@ds.seoultech.ac.kr} \\
  \AND
  Hyunwoo Park\thanks{Corresponding author.} \\
  Graduate School of Data Science\\
  Seoul National University\\
  Seoul, Republic of Korea \\
  \texttt{hyunwoopark@snu.ac.kr} \\
}

\begin{document}

\maketitle

\input{sections/0_abstract}

\input{sections/1_introduction}
\input{sections/2_related_work}

\input{sections/3_preliminary}

\input{sections/4_analysis_and_metrics}

\input{sections/5_method}
\input{sections/6_experiments}

\input{sections/7_conclusion}

\begin{ack}
Use unnumbered first level headings for the acknowledgments. All acknowledgments
go at the end of the paper before the list of references. Moreover, you are required to declare
funding (financial activities supporting the submitted work) and competing interests (related financial activities outside the submitted work).
More information about this disclosure can be found at: \url{https://neurips.cc/Conferences/2026/PaperInformation/FundingDisclosure}.

Do {\bf not} include this section in the anonymized submission, only in the final paper. You can use the \texttt{ack} environment provided in the style file to automatically hide this section in the anonymized submission.
\end{ack}

\clearpage

\medskip

{
\small
\bibliographystyle{plainnat}
\bibliography{references}
}


\appendix

\input{sections/8_appendix}


\newpage

\section*{NeurIPS Paper Checklist}

\begin{enumerate}

\item {\bf Claims}
    \item[] Question: Do the main claims made in the abstract and introduction accurately reflect the paper's contributions and scope?
    \item[] Answer: \answerYes{}
    \item[] Justification: The abstract and introduction (\S\ref{sec:intro}) make four explicit claims, each backed by a corresponding section: (i) state-of-the-art graph-based sign-aware methods are largely valence-blind at the ranking stage despite explicit sign-aware training, and this failure is structurally masked by conventional evaluation metrics, supported by V-AUC diagnosis and linear probes (\S\ref{sec:analysis}); (ii) we propose a family of signed evaluation metrics (Signed Recall, Signed HR, and Signed NDCG) that explicitly penalize the recommendation of disliked items and reduce to conventional metrics at $\gamma\!=\!0$, defined in \S\ref{sec:signed_metrics} with the characterization theorem (Proposition~\ref{prop:characterization}) stated and proved in Appendix~\ref{app:theory}; (iii) we re-evaluate existing state-of-the-art models under the proposed metrics, demonstrating that the established performance landscape fundamentally changes when valence enters the assessment (\S\ref{sec:reeval}); and (iv) we confirm the actionability of the proposed metrics through a proof-of-concept that combines a sign classification auxiliary loss with a valence-aware inference adjustment, showing that models can be guided toward valence-aware behavior without sacrificing conventional relevance (\S\ref{sec:exp}).

\item {\bf Limitations}
    \item[] Question: Does the paper discuss the limitations of the work performed by the authors?
    \item[] Answer: \answerYes{}
    \item[] Justification: The conclusion (\S\ref{sec:conclusion}) discusses two limitations: (i) the analysis focuses on graph-based models using inner-product scoring, leaving generalization to sequential or LLM-based recommender systems for future work; (ii) the proof-of-concept auxiliary loss $\mathcal{L}_\text{sign}$ reveals that enforcing valence separation in dense negative environments can conflict with performance preservation.

\item {\bf Theory assumptions and proofs}
    \item[] Question: For each theoretical result, does the paper provide the full set of assumptions and a complete (and correct) proof?
    \item[] Answer: \answerYes{}
    \item[] Justification: The paper has four formal results, each with assumptions stated within the proposition and a complete proof in the appendix: Proposition~\ref{prop:characterization} (characterization of signed metrics) and Corollary~\ref{cor:monotonicity} (monotonicity in $\gamma$) are proved in Appendix~\ref{app:theory}; Propositions~\ref{prop:pn_coupling} and~\ref{prop:sign_decoupling} (gradient coupling and decoupling) are proved in Appendix~\ref{app:gradient}. All theorems are numbered and cross-referenced.

    \item {\bf Experimental result reproducibility}
    \item[] Question: Does the paper fully disclose all the information needed to reproduce the main experimental results of the paper to the extent that it affects the main claims and/or conclusions of the paper (regardless of whether the code and data are provided or not)?
    \item[] Answer: \answerYes{}
    \item[] Justification: We provide the full experimental specification: dataset preprocessing (\S\ref{sec:exp}, Table~\ref{tab:datasets}), evaluation protocol (7:1:2 split, 5-core filtering, three seeds $\{42, 2024, 2025\}$, the Both P\&N user-pool restriction with subset-validity audit in Appendix~\ref{app:subset_validation}), per-model hyperparameters and training details (Appendix~\ref{app:impl}, Table~\ref{tab:hyperparams}), two dataset-specific overrides (Appendix~\ref{app:impl_overrides}), and the $\mathcal{L}_\text{sign}$ search grid and inference-time scoring (Appendix~\ref{app:impl_lsign}, Eq.~\ref{eq:final_score}). The values in Table~\ref{tab:hyperparams} are exactly the configurations used in our experiments and present in the released code/scripts; running the released scripts reproduces the reported numbers. Shared infrastructure choices (embedding dimension, evaluation protocol, seeds, splits) are standardized across methods for fair comparison; per-method values follow each baseline's released-code defaults.

\item {\bf Open access to data and code}
    \item[] Question: Does the paper provide open access to the data and code, with sufficient instructions to faithfully reproduce the main experimental results, as described in supplemental material?
    \item[] Answer: \answerYes{}
    \item[] Justification: All datasets used (Amazon-CDs, Amazon-Music, Epinions, KuaiRand, KuaiRec) are publicly available. We release code, preprocessing scripts, per-seed user-level scores, and signed-metric implementations at the anonymized link in the abstract. The signed-metric toolkit takes as input the per-user top-$K$ ranked item lists together with the $(\mathcal{P}_u, \mathcal{N}_u)$ labels, making it a drop-in addition regardless of the underlying architecture or scoring function.
    \item[] Guidelines:
    \begin{itemize}
        \item The answer NA means that paper does not include experiments requiring code.
        \item Please see the NeurIPS code and data submission guidelines (\url{https://nips.cc/public/guides/CodeSubmissionPolicy}) for more details.
        \item While we encourage the release of code and data, we understand that this might not be possible, so ``No'' is an acceptable answer. Papers cannot be rejected simply for not including code, unless this is central to the contribution (e.g., for a new open-source benchmark).
        \item The instructions should contain the exact command and environment needed to run to reproduce the results. See the NeurIPS code and data submission guidelines (\url{https://nips.cc/public/guides/CodeSubmissionPolicy}) for more details.
        \item The authors should provide instructions on data access and preparation, including how to access the raw data, preprocessed data, intermediate data, and generated data, etc.
        \item The authors should provide scripts to reproduce all experimental results for the new proposed method and baselines. If only a subset of experiments are reproducible, they should state which ones are omitted from the script and why.
        \item At submission time, to preserve anonymity, the authors should release anonymized versions (if applicable).
        \item Providing as much information as possible in supplemental material (appended to the paper) is recommended, but including URLs to data and code is permitted.
    \end{itemize}

\item {\bf Experimental setting/details}
    \item[] Question: Does the paper specify all the training and test details (e.g., data splits, hyperparameters, how they were chosen, type of optimizer, etc.) necessary to understand the results?
    \item[] Answer: \answerYes{}
    \item[] Justification: Common configuration (embedding dimension $d\!=\!64$, Adam optimizer with $\beta_1\!=\!0.9, \beta_2\!=\!0.999$, three seeds, 7:1:2 splits, 5-core filtering, label thresholds per dataset) is given in Appendix~\ref{app:impl_common}. Per-model hyperparameters (learning rate, layers, batch size, epochs, patience, init) are tabulated in Table~\ref{tab:hyperparams}; the values shown are exactly those used in our experiments and present in the released code/scripts. Two dataset-specific overrides are documented in Appendix~\ref{app:impl_overrides}. Per-method hyperparameters follow each baseline's released-code defaults; shared infrastructure choices listed above are standardized across methods for fair comparison.

\item {\bf Experiment statistical significance}
    \item[] Question: Does the paper report error bars suitably and correctly defined or other appropriate information about the statistical significance of the experiments?
    \item[] Answer: \answerYes{}
    \item[] Justification: Mean$\pm$std (one standard deviation) over three random seeds is reported throughout (Tables~\ref{tab:full_diagnosis}, \ref{tab:probe}, \ref{tab:reeval}, \ref{tab:lsign}, \ref{tab:neg_infiltration}). Significance markers in Table~\ref{tab:reeval} ($\ddagger, \dagger$) are based on paired $t$-tests on user-level scores with Holm--Bonferroni correction within each column at family-wise $\alpha\!=\!0.05$; seed averaging precedes testing to avoid independence violations from per-seed pooling. Full $B/W$ counts are reported in Appendix~\ref{app:sig_test} (Table~\ref{tab:sig_all}) on three datasets spanning low, medium, and high $\rho$. Subset-induced ranking distortion is separately audited via Spearman/Kendall rank correlation (Appendix~\ref{app:subset_validation}, Table~\ref{tab:subset_validation}).

\item {\bf Experiments compute resources}
    \item[] Question: For each experiment, does the paper provide sufficient information on the computer resources (type of compute workers, memory, time of execution) needed to reproduce the experiments?
    \item[] Answer: \answerYes{}
    \item[] Justification: Per-method wall-clock times grouped into Fast ($<\!30$ min), Medium (1--3 h), and Slow ($\geq\!2$ h) on a single 24GB GPU are documented in Appendix~\ref{app:impl_compute}. Total compute is approximately 140 GPU-hours across $10\,\text{methods} \times 5\,\text{datasets} \times 3\,\text{seeds}$. The exception (NFARec on Amazon-CDs requiring $\geq\!80$GB memory due to its $46{,}464^2$ correlation matrix at $d_\text{model}=1024$) is explicitly flagged with a smaller-hardware substitute ($d_\text{model}=512$) and the resulting $\sim\!8\%$ NDCG degradation noted.

\item {\bf Code of ethics}
    \item[] Question: Does the research conducted in the paper conform, in every respect, with the NeurIPS Code of Ethics \url{https://neurips.cc/public/EthicsGuidelines}?
    \item[] Answer: \answerYes{}
    \item[] Justification: We have read the NeurIPS Code of Ethics and the research conforms with it. The paper uses publicly available recommendation datasets, conducts no human-subject research, and does not release any new dataset, model, or scraped content that could pose misuse risk.

\item {\bf Broader impacts}
    \item[] Question: Does the paper discuss both potential positive societal impacts and negative societal impacts of the work performed?
    \item[] Answer: \answerYes{}
    \item[] Justification: The intended societal benefit, reducing the recommendation of explicitly disliked content (filter-bubble narrowing, user reactance, dissatisfaction harms), is the central motivation of the paper, with citations to user-experience evidence (\S\ref{sec:intro}). Potential negative impact: like any evaluation criterion, signed metrics can be over-optimized (e.g., over-conservative recommendations that suppress exploration); we mitigate this by exposing $\gamma$ as a tunable evaluation-time dial and recommending reporting at multiple $\gamma$ rather than committing to a single value (Appendix~\ref{app:gamma_choice}).

\item {\bf Safeguards}
    \item[] Question: Does the paper describe safeguards that have been put in place for responsible release of data or models that have a high risk for misuse (e.g., pretrained language models, image generators, or scraped datasets)?
    \item[] Answer: \answerNA{}
    \item[] Justification: The paper releases evaluation metrics and a small auxiliary loss as a proof of concept, neither of which constitutes a high-misuse-risk asset (no pretrained language model, image generator, or scraped dataset).

\item {\bf Licenses for existing assets}
    \item[] Question: Are the creators or original owners of assets (e.g., code, data, models), used in the paper, properly credited and are the license and terms of use explicitly mentioned and properly respected?
    \item[] Answer: \answerYes{}
    \item[] Justification: All datasets (Amazon Reviews, Epinions, KuaiRand, KuaiRec) and baseline implementations (LightGCN, LightGCL, XSimGCL, GFormer, SiReN, SiGRec, NFARec, SIGformer, Pone-GNN, SDCGCL) are properly cited in \S\ref{sec:exp} and \S\ref{sec:related}. We use these assets under their respective public terms; we did not modify or redistribute the original datasets, only computed evaluation metrics on top of model outputs. License information for each asset (datasets and baseline implementations) is documented in Section 9 (Licenses) of the README in our released code repository.

\item {\bf New assets}
    \item[] Question: Are new assets introduced in the paper well documented and is the documentation provided alongside the assets?
    \item[] Answer: \answerYes{}
    \item[] Justification: We release the signed-metric toolkit (Signed Recall, Signed HR, Signed NDCG implementations) and the reproducibility package (preprocessing scripts, training scripts, per-seed user-level scores) at the anonymized link in the abstract. Documentation is included in the repository; the toolkit accepts the per-user top-$K$ ranked item lists and the $(\mathcal{P}_u, \mathcal{N}_u)$ labels, with the metric definitions specified in \S\ref{sec:signed_metrics}.

\item {\bf Crowdsourcing and research with human subjects}
    \item[] Question: For crowdsourcing experiments and research with human subjects, does the paper include the full text of instructions given to participants and screenshots, if applicable, as well as details about compensation (if any)?
    \item[] Answer: \answerNA{}
    \item[] Justification: The paper involves no crowdsourcing or human-subject research. All experiments use existing public datasets of past user interactions.

\item {\bf Institutional review board (IRB) approvals or equivalent for research with human subjects}
    \item[] Question: Does the paper describe potential risks incurred by study participants, whether such risks were disclosed to the subjects, and whether Institutional Review Board (IRB) approvals (or an equivalent approval/review based on the requirements of your country or institution) were obtained?
    \item[] Answer: \answerNA{}
    \item[] Justification: No human-subject research was conducted; IRB approval is not applicable.

\item {\bf Declaration of LLM usage}
    \item[] Question: Does the paper describe the usage of LLMs if it is an important, original, or non-standard component of the core methods in this research? Note that if the LLM is used only for writing, editing, or formatting purposes and does not impact the core methodology, scientific rigorousness, or originality of the research, declaration is not required.
    \item[] Answer: \answerNA{}
    \item[] Justification: LLMs were used only for incidental writing assistance (grammar, phrasing). No LLM is part of the core methodology, experimental pipeline, or scientific analysis.
    \item[] Guidelines:
    \begin{itemize}
        \item The answer NA means that the core method development in this research does not involve LLMs as any important, original, or non-standard components.
        \item Please refer to our LLM policy (\url{https://neurips.cc/Conferences/2026/LLM}) for what should or should not be described.
    \end{itemize}

\end{enumerate}

\end{document}

%% file: sections/0_abstract.tex
\begin{abstract}
    Sign-aware recommender systems have recently been developed to leverage negative feedback for a deeper understanding of user preferences.
    However, our empirical diagnosis reveals that state-of-the-art graph-based sign-aware recommender systems are paradoxically \emph{valence-blind}.
    Even though they explicitly incorporate sign information during training, they consistently fail to differentiate liked items from disliked ones at the ranking stage, frequently infiltrating top-$K$ recommendations with disliked content.
    Through linear probing, we show that while valence information exists in the learned embeddings, it remains inaccessible to the inner-product scoring function.
    This widespread failure remains entirely undetected because conventional evaluation metrics, such as Recall, HR, and NDCG, assign a uniform utility of zero to both negative and unobserved items, creating a systematic evaluation blind spot.
    To bridge this gap, we propose a family of signed metrics, Signed Recall, Signed HR, and Signed NDCG, that explicitly penalize the recommendation of disliked content.
    Systematic re-evaluation under our proposed metrics fundamentally reshapes the established performance landscape, revealing that methods ranked highly under conventional metrics often fail to protect users from disliked content.
    Finally, through a proof-of-concept auxiliary loss, we confirm that the proposed metrics provide actionable training signals, guiding models toward valence-aware behavior without sacrificing conventional relevance.
    For transparency, our source code is available at: \url{https://anonymous.4open.science/r/signed-rec-benchmark-07E4}
\end{abstract}

%% file: sections/1_introduction.tex
\section{Introduction}
\label{sec:intro}

Sign-aware recommender systems have recently emerged as a promising direction by explicitly incorporating negative user feedback, such as low ratings, dislikes, and skips, alongside positive interactions to construct signed graphs that distinguish the nature of each signal~\cite{siren, sigrec, nfarec, sigformer, ponegnn, sdcgcl}.
Unlike conventional graph-based recommender systems that treat all observed interactions as positive signals~\cite{lightgcn, gformer, lightgcl, xsimgcl}, these methods are designed with the explicit goal of distinguishing the valence of user interactions, treating positive and negative feedback as carrying fundamentally different preference signals.

However, our empirical analysis reveals a paradox.
Even though sign-aware methods explicitly incorporate sign information during training, they consistently fail to differentiate liked items from disliked ones at the ranking stage.
As illustrated in Figure~\ref{fig:score_dist}, the predicted score distributions for positive and negative items overlap almost entirely and, in this specific case, the negative distribution is shifted slightly higher, indicating the model assigns marginally greater preference scores to disliked items than to liked items.
We quantify this failure through Valence AUC (V-AUC), which measures the probability that a model ranks a positive item above a negative item for the same user.
Across four of five benchmarks, existing sign-aware methods achieve V-AUC near random chance, indicating that their ranking stage largely fails to reflect the valence of observed interactions.

To understand why this happens, we apply linear probes~\cite{alain2017probing} to the learned embeddings and find that valence information is in fact present in the embedding space, yet remains inaccessible to the inner-product scoring function at inference.
This gap between what the embeddings encode and what the scoring function can extract points to a structural bottleneck inherent to inner-product-based ranking, rather than a failure of the sign-aware training objective.
This bottleneck is particularly pronounced in architecture-driven methods, where sophisticated sign-aware encodings are ultimately reduced to a single inner product score at the ranking stage.
The problem is further compounded by the evaluation environment itself.
Since model selection and hyperparameter tuning are driven entirely by conventional metrics, any valence-oriented signal that conflicts with relevance is suppressed during training, leaving this misalignment unaddressed even in models that have enough capacity to encode valence.
In short, valence cannot be managed unless it is measured.

\begin{wrapfigure}{r}{0.42\linewidth}
\vspace{-1em}
\centering
\includegraphics[width=\linewidth]{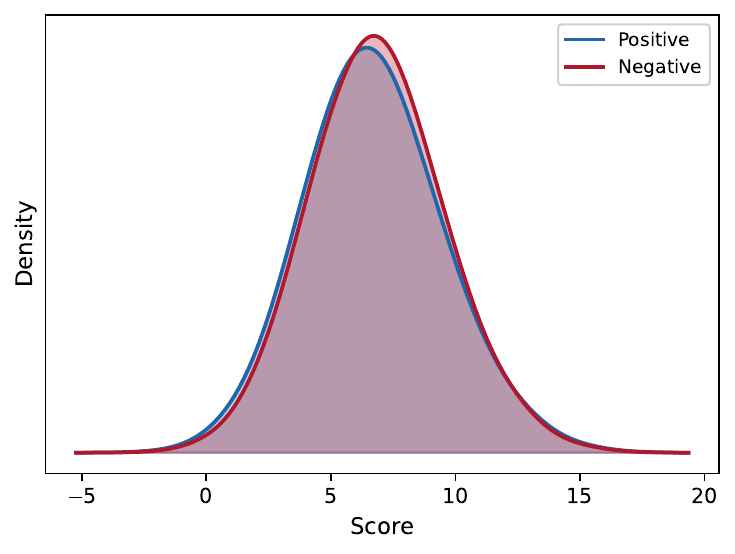}
\caption{Density of predicted scores for positive (blue) and negative (red) test items. Results are shown for SIGformer, a representative sign-aware model, evaluated on the Amazon-CDs benchmark dataset.}
\label{fig:score_dist}
\vspace{-1em}
\end{wrapfigure}

This misalignment goes undetected under conventional evaluation protocols.
Metrics such as Recall, HR, and NDCG assign a uniform utility of zero to both negative and unobserved items, making them structurally incapable of penalizing the recommendation of disliked content.
As a result, a model that frequently places disliked items in the top-$K$ recommendation list receives the same evaluation score as one that avoids them, masking a systematic degradation of user experience~\cite{ma2022negative, benporat2022departing, fitzsimons2004reactance}.

Motivated by these findings, we propose a family of signed evaluation metrics, Signed Recall, Signed HR, and Signed NDCG, that explicitly penalize the recommendation of disliked content by extending conventional metrics with a tunable penalty parameter $\gamma$.
At $\gamma=0$, the proposed metrics reduce exactly to their conventional counterparts, ensuring backward compatibility with previously reported results.
Re-evaluating existing state-of-the-art models under these metrics fundamentally reshapes the established performance landscape, revealing that methods ranked highly under conventional metrics often fail to protect users from disliked content.
Through a proof-of-concept auxiliary loss, we additionally verify that the proposed metrics provide actionable training signals, guiding models toward valence-aware behavior without sacrificing conventional relevance performance.

In summary, the main contributions of this work are as follows:
\begin{itemize}
    \item We diagnose that state-of-the-art graph-based sign-aware methods are largely valence-blind at the ranking stage despite explicit sign-aware training, and demonstrate that this failure is structurally masked by conventional evaluation metrics.
    \item We propose a family of signed evaluation metrics, Signed Recall, Signed HR, and Signed NDCG, that explicitly penalize the recommendation of disliked items and reduce to conventional metrics at $\gamma=0$.
    \item We re-evaluate existing state-of-the-art models under the proposed metrics, demonstrating that the established performance landscape fundamentally changes when valence enters the assessment.
    \item We confirm the actionability of the proposed metrics through a proof-of-concept that combines a sign classification auxiliary loss with a valence-aware inference adjustment, showing that models can be guided toward valence-aware behavior without sacrificing conventional relevance.
\end{itemize}

%% file: sections/2_related_work.tex
\section{Related Work}
\label{sec:related}

\paragraph{Sign-aware Graph-based Recommender Systems} Early works like SiReN~\cite{siren} and SiGRec~\cite{sigrec} pioneered this field by learning separate representations for positive and negative interactions via GNNs (graph neural networks) and sign-aware BPR (Bayesian personalized ranking) objectives.
NFARec~\cite{nfarec} employs feedback-aware hypergraph convolutions, while SIGformer~\cite{sigformer} leverages sign-aware spectral encodings within a Transformer architecture.
Most recently, Pone-GNN~\cite{ponegnn} disentangles user interests and disinterests, and SDCGCL~\cite{sdcgcl} utilizes dual-channel graph contrastive learning.
Despite these remarkable architectural advancements in leveraging negative feedback, a critical misalignment persists in that these sign-aware models are evaluated using conventional, sign-blind metrics, leaving their valence-aware behavior largely unexamined in the literature.

%% file: sections/3_preliminary.tex
\section{Preliminaries}
\label{sec:prelim}

\subsection{Problem Formulation}
We consider a recommender system with a user set $\mathcal{U}$ and an item set $\mathcal{I}$, where $|\mathcal{U}|=n$ and $|\mathcal{I}|=m$.
The user-item interactions are represented as $\mathcal{D}=\{(u, i, y_{ui})\}$, where the label $y_{ui}\in\{+1, -1, ?\}$ denotes a user's positive, negative, or unobserved interaction, respectively.
For a given user $u$, we partition the items into positive($\mathcal{P}_u$), negative($\mathcal{N}_u$), and unobserved($\mathcal{U}_u$).
This interaction data naturally forms a signed bipartite graph $\mathcal{G}=(\mathcal{V}, \mathcal{E}^+, \mathcal{E}^-)$, where the node set is $\mathcal{V}=\mathcal{U}\cup\mathcal{I}$ and the signed edges are defined as $\mathcal{E}^{\pm}=\{(u,i):y_{ui}=\pm1\}$.

To capture user preferences, sign-aware recommender systems learn $d$-dimensional representations, $\mathbf{e}_u, \mathbf{e}_i\in\mathbb{R}^d$, for each user and item by encoding the signed graph $\mathcal{G}$.
The predicted preference score is then computed via inner product, $\hat{y}_{ui}=\mathbf{e}_u^\top\mathbf{e}_i$, which is subsequently used to rank items and generate the top-$K$ recommendation list.
To optimize these representations, sign-aware methods often extend the original BPR into a sign-aware BPR objective.
While specific formulations vary, a representative approach is the signed extension proposed by SIGformer \cite{sigformer}, defined as:
\begin{equation}
\mathcal{L}=-\sum_{(u,i)\in\mathcal{E}^+}\ln\sigma\left(\hat{y}_{ui}-\hat{y}_{uj}\right) -\sum_{(u,i)\in\mathcal{E}^-}\ln\sigma\left(\beta(\hat{y}_{ui}-\hat{y}_{uj})\right),
\end{equation}
where for each positive or negative feedback $(u,i)$, the model samples an item $j\in\mathcal{U}_u$ that the user has not interacted with, and $\sigma(\cdot)$ is the sigmoid function.
Crucially, the hyperparameter $\beta\in[-1, 1]$ controls the influence from the negative feedback.

\subsection{Conventional Evaluation Metrics}
To assess the quality of the generated recommendations, standard evaluation protocols rely on a suite of canonical metrics.
Given a user $u$ and a ranked top-$K$ recommendation list $\pi_K(u)$, these metrics evaluate performance based exclusively on a binary relevance indicator, defined as $\mathrm{rel}_u(i)=\mathbb{I}[i\in\mathcal{P}_u]$, where $\mathbb{I}(\cdot)$ is the indicator function.
The three canonical metrics, Recall@$K$, HR@$K$, NDCG@$K$ are formulated as follows:
\begin{align}
\mathrm{Recall}@K=\frac{|\mathcal{P}_u\cap\pi_K(u)|}{|\mathcal{P}_u|}, \quad
\mathrm{HR}@K=\mathbb{I}[|\mathcal{P}_u\cap\pi_K(u)|>0] \\
\mathrm{NDCG}@K=\frac{\mathrm{DCG}@K}{\mathrm{IDCG}@K}, \quad
\mathrm{DCG}@K=\sum_{k=1}^K\frac{\mathrm{rel}_u(\pi_K(u)[k])}{\log_2(k+1)}
\end{align}
Here, $\pi_K(u)[k]$ denotes the item ranked at the $k$-th position, and $\mathrm{IDCG}@K$ represents the ideal DCG obtained by perfectly ranking all positive items.
While these metrics effectively measure the retrieval of relevant items, their relevance-centric formulation exposes a critical structural limitation in the sign-aware setting.
They assign the exact same utility score of zero to both explicitly negative items ($\mathcal{N}_u$) and unobserved items ($\mathcal{U}_u$).
Consequently, this evaluation paradigm remains completely sign-blind, strictly rewarding the retrieval of positive content while failing to penalize the recommendation of explicitly disliked items.
This sign-blindness is not merely a theoretical concern.
In \S\ref{sec:analysis}, we demonstrate empirically that it masks a systematic failure of sign-aware methods to separate liked from disliked items at the ranking stage.

%% file: sections/4_analysis_and_metrics.tex
\section{Empirical Diagnosis of Valence-Blindness}
\label{sec:analysis}

In this section, we investigate whether existing sign-aware models truly differentiate positive from negative items during inference, and uncover the structural causes behind their limitations.
Specifically, we introduce V-AUC as a diagnostic for positive-versus-negative separation, demonstrating that state-of-the-art methods widely fail this test across multiple benchmarks (\S\ref{sec:v_auc}).
We then scrutinize the root cause of this failure, arguing that it reflects a critical misalignment between valence-oriented training signals and inner-product scoring, heavily amplified by the relevance-oriented evaluation environment (\S\ref{sec:inner_prod}).

\subsection{Are Sign-aware Methods Really Sign-aware?}
\label{sec:v_auc}
As discussed in \S\ref{sec:prelim}, conventional metrics solely evaluate whether a model ranks explicitly interacted items $(\mathcal{P}_u\cup\mathcal{N}_u)$ higher than unobserved ones $(\mathcal{U}_u)$.
However, true sign-awareness requires a model to prioritize positive interactions over negative ones within the observed set.

\begin{table*}[b!]
\caption{Full diagnostic analysis of sign-aware methods across five benchmarks ($K=20$). $\rho=|\mathcal{E}^-|/|\mathcal{E}^+|$ denotes the ratio of negative to positive interactions. ``V-AUC'' is the probability of ranking a positive item above a negative one. ``Overlap'' is the fraction of users whose positive and negative score distributions overlap. Methods are grouped by where sign-awareness is embedded: \textbf{Architecture-driven} (embed sign in model structure with standard or auxiliary-only loss) vs.\ \textbf{Loss-driven} (modify the training objective). Within each group, methods are ordered chronologically. Mean\textsubscript{$\pm$std} over three seeds. \textbf{Bold}: best, \underline{underline}: second best per benchmark.}
\label{tab:full_diagnosis}
\centering
\footnotesize
\setlength{\tabcolsep}{4pt}
\renewcommand{\arraystretch}{0.9}
\begin{tabular}{l ll ccc}
\toprule
\textbf{Dataset} ($\rho$) & \textbf{Group} & \textbf{Method} & \textbf{NDCG}$(\uparrow)$ & \textbf{V-AUC}$(\uparrow)$ & \textbf{Overlap}$(\downarrow)$ \\
\midrule
\multirow{6}{*}{Amazon-CDs (0.22)}
& \multirow{3}{*}{Arch.-driven}
  & SiGRec     & .0506\textsubscript{$\pm$.0024} & .527\textsubscript{$\pm$.030} & 73.9\textsubscript{$\pm$.6}\% \\
& & NFARec     & \underline{.0840\textsubscript{$\pm$.0008}} & .553\textsubscript{$\pm$.000} & 74.6\textsubscript{$\pm$.2}\% \\
& & SIGformer  & \textbf{.0865\textsubscript{$\pm$.0019}} & .486\textsubscript{$\pm$.001} & 75.4\textsubscript{$\pm$.1}\% \\
\cmidrule(lr){2-6}
& \multirow{3}{*}{Loss-driven}
  & SiReN      & .0810\textsubscript{$\pm$.0028} & .560\textsubscript{$\pm$.001} & \textbf{71.6\textsubscript{$\pm$.4}\%} \\
& & Pone-GNN   & .0641\textsubscript{$\pm$.0047} & \underline{.565\textsubscript{$\pm$.001}} & \textbf{71.6\textsubscript{$\pm$.0}\%} \\
& & SDCGCL     & .0550\textsubscript{$\pm$.0218} & \textbf{.588\textsubscript{$\pm$.026}} & \underline{71.9\textsubscript{$\pm$.8}\%} \\
\midrule
\multirow{6}{*}{Amazon-Music (0.25)}
& \multirow{3}{*}{Arch.-driven}
  & SiGRec     & .0975\textsubscript{$\pm$.0367} & \textbf{.616\textsubscript{$\pm$.065}} & 68.6\textsubscript{$\pm$6.6}\% \\
& & NFARec     & .1026\textsubscript{$\pm$.0013} & .533\textsubscript{$\pm$.004} & 73.0\textsubscript{$\pm$1.7}\% \\
& & SIGformer  & \textbf{.1766\textsubscript{$\pm$.0069}} & .486\textsubscript{$\pm$.002} & 71.5\textsubscript{$\pm$.6}\% \\
\cmidrule(lr){2-6}
& \multirow{3}{*}{Loss-driven}
  & SiReN      & \underline{.1683\textsubscript{$\pm$.0041}} & .542\textsubscript{$\pm$.010} & 68.6\textsubscript{$\pm$1.1}\% \\
& & Pone-GNN   & .1566\textsubscript{$\pm$.0055} & .595\textsubscript{$\pm$.004} & \textbf{67.2\textsubscript{$\pm$.7}\%} \\
& & SDCGCL     & .1642\textsubscript{$\pm$.0007} & \underline{.607\textsubscript{$\pm$.006}} & \underline{67.6\textsubscript{$\pm$.7}\%} \\
\midrule
\multirow{6}{*}{Epinions (0.37)}
& \multirow{3}{*}{Arch.-driven}
  & SiGRec     & .0410\textsubscript{$\pm$.0004} & .487\textsubscript{$\pm$.011} & 82.5\textsubscript{$\pm$.3}\% \\
& & NFARec     & .0501\textsubscript{$\pm$.0004} & .480\textsubscript{$\pm$.001} & 83.3\textsubscript{$\pm$.1}\% \\
& & SIGformer  & \textbf{.0615\textsubscript{$\pm$.0023}} & .458\textsubscript{$\pm$.002} & 83.6\textsubscript{$\pm$.2}\% \\
\cmidrule(lr){2-6}
& \multirow{3}{*}{Loss-driven}
  & SiReN      & .0491\textsubscript{$\pm$.0009} & .500\textsubscript{$\pm$.001} & 80.8\textsubscript{$\pm$.0}\% \\
& & Pone-GNN   & \underline{.0540\textsubscript{$\pm$.0005}} & \underline{.523\textsubscript{$\pm$.000}} & \underline{80.1\textsubscript{$\pm$.2}\%} \\
& & SDCGCL     & .0534\textsubscript{$\pm$.0058} & \textbf{.549\textsubscript{$\pm$.008}} & \textbf{78.2\textsubscript{$\pm$.5}\%} \\
\midrule
\multirow{6}{*}{KuaiRand (1.25)}
& \multirow{3}{*}{Arch.-driven}
  & SiGRec     & .0495\textsubscript{$\pm$.0024} & .530\textsubscript{$\pm$.004} & 74.5\textsubscript{$\pm$.3}\% \\
& & NFARec     & .0463\textsubscript{$\pm$.0006} & .498\textsubscript{$\pm$.003} & 74.8\textsubscript{$\pm$.1}\% \\
& & SIGformer  & \textbf{.0680\textsubscript{$\pm$.0014}} & .538\textsubscript{$\pm$.002} & 74.2\textsubscript{$\pm$.2}\% \\
\cmidrule(lr){2-6}
& \multirow{3}{*}{Loss-driven}
  & SiReN      & .0374\textsubscript{$\pm$.0007} & \textbf{.627\textsubscript{$\pm$.005}} & \textbf{68.2\textsubscript{$\pm$.2}\%} \\
& & Pone-GNN   & \underline{.0516\textsubscript{$\pm$.0025}} & .528\textsubscript{$\pm$.015} & 71.7\textsubscript{$\pm$.4}\% \\
& & SDCGCL     & .0428\textsubscript{$\pm$.0041} & \underline{.545\textsubscript{$\pm$.008}} & \underline{69.5\textsubscript{$\pm$.4}\%} \\
\midrule
\multirow{6}{*}{KuaiRec (5.95)}
& \multirow{3}{*}{Arch.-driven}
  & SiGRec     & \underline{.0412\textsubscript{$\pm$.0011}} & \textbf{.921\textsubscript{$\pm$.004}} & \underline{69.5\textsubscript{$\pm$.4}\%} \\
& & NFARec     & .0007\textsubscript{$\pm$.0001} & .168\textsubscript{$\pm$.016} & 100.0\% \\
& & SIGformer$^\dagger$ & \textbf{.0496\textsubscript{$\pm$.0014}} & \underline{.891\textsubscript{$\pm$.015}} & \textbf{67.8\textsubscript{$\pm$1.9}\%} \\
\cmidrule(lr){2-6}
& \multirow{3}{*}{Loss-driven}
  & SiReN      & .0153\textsubscript{$\pm$.0225} & .463\textsubscript{$\pm$.244} & 99.6\textsubscript{$\pm$.7}\% \\
& & Pone-GNN   & .0314\textsubscript{$\pm$.0062} & .677\textsubscript{$\pm$.002} & 98.8\textsubscript{$\pm$.7}\% \\
& & SDCGCL     & .0382\textsubscript{$\pm$.0008} & .874\textsubscript{$\pm$.004} & 82.5\textsubscript{$\pm$.4}\% \\
\bottomrule
\multicolumn{6}{l}{\footnotesize $^\dagger$ SIGformer uses $\beta = -0.2$ on KuaiRec, which reverses the negative feedback direction.}
\end{tabular}
\end{table*}

To strictly quantify this discriminative capability, we introduce a diagnostic metric, V-AUC.
For a given user $u$ with $\mathcal{P}_u, \mathcal{N}_u\neq\emptyset$, V-AUC isolates the evaluation to explicitly interacted items, measuring the probability that the model assigns a higher preference score to a positive item ($p\in\mathcal{P}_u$) than to a negative item ($n\in\mathcal{N}_u$):
\begin{equation}
    \text{V-AUC}(u)=\frac{1}{|\mathcal{P}_u||\mathcal{N}_u|}\sum_{p\in\mathcal{P}_u}\sum_{n\in\mathcal{N}_u}\mathbb{I}[\hat{y}_{up}>\hat{y}_{un}]
\end{equation}

A V-AUC score of 1.0 indicates perfect separation, while 0.5 signifies that the scoring function is entirely valence-blind, behaving no better than a random guess.
Alongside V-AUC, we quantify this failure via the \emph{Overlap} ratio, defined as the fraction of users whose predicted score distributions for positive and negative items entangle.
Formally, this occurs when $\min_{p\in\mathcal{P}_u}\hat{y}_{up}<\max_{n\in\mathcal{N}_u}\hat{y}_{un}$.

Table \ref{tab:full_diagnosis} reports these diagnostic results across five real-world benchmarks (comprehensive dataset configurations are deferred to \S\ref{sec:exp}.)
While sign-aware methods achieve competitive NDCG scores, their V-AUC scores are close to 0.5.
Paradoxically, the highest-NDCG models frequently exhibit the worst valence separation.
For example, SIGformer dominates NDCG on the Amazon datasets and Epinions, but records V-AUC scores below 0.5, actively ranking disliked items above liked ones.

\begin{table}[b!]
\caption{Linear probe AUC for sign classification on learned embeddings ($K\!=\!20$). A logistic regression probe on $\mathbf{e}_u \odot \mathbf{e}_i$ (5-fold CV) reveals how much valence information exists in the embedding space, independent of whether the inner-product scoring function can extract it. Signed methods are further grouped by where sign-awareness is embedded: \textbf{Arch.-driven} (sign in model structure) vs.\ \textbf{Loss-driven} (sign in training objective). Within each group, methods are ordered chronologically. Results reported as mean\textsubscript{$\pm$std} over three seeds. \textbf{Bold}: best, \underline{underline}: second best per column.}
\label{tab:probe}
\centering
\footnotesize
\begin{tabular}{ll l ccccc}
\toprule
& & \textbf{Method} & \textbf{\makecell{Amazon-CDs}} & \textbf{\makecell{Amazon-Music}} & \textbf{\makecell{Epinion}} & \textbf{\makecell{KuaiRand}} & \textbf{\makecell{KuaiRec}} \\
\midrule
\multirow{4}{*}{\rotatebox{90}{Unsigned}}
& & LightGCN    & .565\textsubscript{$\pm$.003} & .549\textsubscript{$\pm$.009} & .555\textsubscript{$\pm$.003} & .560\textsubscript{$\pm$.029} & .831\textsubscript{$\pm$.034} \\
& & LightGCL    & .550\textsubscript{$\pm$.004} & .523\textsubscript{$\pm$.008} & .538\textsubscript{$\pm$.004} & .533\textsubscript{$\pm$.026} & .822\textsubscript{$\pm$.002} \\
& & XSimGCL     & .558\textsubscript{$\pm$.001} & .533\textsubscript{$\pm$.006} & .533\textsubscript{$\pm$.008} & .532\textsubscript{$\pm$.004} & .480\textsubscript{$\pm$.042} \\
& & GFormer     & .604\textsubscript{$\pm$.003} & .564\textsubscript{$\pm$.016} & .581\textsubscript{$\pm$.006} & .601\textsubscript{$\pm$.006} & .869\textsubscript{$\pm$.005} \\
\midrule
\multirow{6}{*}{\rotatebox{90}{Signed}}
& \multirow{3}{*}{Arch.}
  & SiGRec      & \textbf{.696\textsubscript{$\pm$.002}} & \textbf{.674\textsubscript{$\pm$.008}} & \textbf{.732\textsubscript{$\pm$.005}} & \textbf{.639\textsubscript{$\pm$.004}} & \underline{.911\textsubscript{$\pm$.006}} \\
& & NFARec      & \underline{.642\textsubscript{$\pm$.001}} & .556\textsubscript{$\pm$.009} & \underline{.631\textsubscript{$\pm$.006}} & .582\textsubscript{$\pm$.003} & .832\textsubscript{$\pm$.003} \\
& & SIGformer   & .561\textsubscript{$\pm$.011} & .559\textsubscript{$\pm$.019} & .584\textsubscript{$\pm$.001} & .593\textsubscript{$\pm$.010} & .893\textsubscript{$\pm$.013} \\
\cmidrule(lr){2-8}
& \multirow{3}{*}{Loss}
  & SiReN       & .602\textsubscript{$\pm$.020} & .556\textsubscript{$\pm$.016} & .565\textsubscript{$\pm$.011} & \underline{.630\textsubscript{$\pm$.004}} & .746\textsubscript{$\pm$.123} \\
& & Pone-GNN    & .563\textsubscript{$\pm$.005} & \underline{.596\textsubscript{$\pm$.003}} & .555\textsubscript{$\pm$.007} & .603\textsubscript{$\pm$.006} & .859\textsubscript{$\pm$.025} \\
& & SDCGCL      & .593\textsubscript{$\pm$.031} & .594\textsubscript{$\pm$.007} & .554\textsubscript{$\pm$.010} & .540\textsubscript{$\pm$.012} & \textbf{.924\textsubscript{$\pm$.002}} \\
\bottomrule
\end{tabular}%
\end{table}

\subsection{Why Does This Happen?}
\label{sec:inner_prod}
The observed results in \S\ref{sec:v_auc} raise a critical question about whether sign-aware methods fail because the training signal is insufficient to encode valence, or because the learned representations contain valence information that the scoring function cannot access. To answer this, we apply linear probes~\cite{alain2017probing} to the learned embeddings.
Specifically, we train a logistic regression classifier on the element-wise embedding product $\mathbf{e}_u\odot\mathbf{e}_i$ to predict whether an interaction is positive or negative to the user, and check the probing AUC in Table \ref{tab:probe}.

Since the probe uses element-wise product rather than inner product, a high probe AUC alongside a low V-AUC directly indicates that valence information exists in the embeddings but is inaccessible to inner-product scoring.
Comparing Table~\ref{tab:probe} with the V-AUC results in Table~\ref{tab:full_diagnosis} reveals a consistent gap across all methods.
For instance, SiGRec achieves a probe AUC of 0.696 on Amazon-CDs, indicating that its embeddings already carry meaningful valence information.
Yet its V-AUC remains at 0.527, barely above random chance.
This gap confirms that the failure is not one of learning but of information extraction from the learned embeddings.
While embedding space encodes valence, the inner-product $\hat{y}_{ui}=\mathbf{e}_u^\top\mathbf{e}_i$ cannot read it, which is consistent with the known limitations of signed-graph theory~\cite{sgcn, sbgnn, sbgcl}.
This pattern is particularly evident in architecture-driven methods, where sophisticated sign-aware encodings are ultimately neutralized at the scoring stage despite their embeddings carrying substantial valence information.

%% file: sections/5_method.tex
\section{The Proposed Signed Evaluation Metrics}
\label{sec:signed_metrics}

To bridge the evaluation gap diagnosed in \S\ref{sec:analysis}, we extend conventional binary relevance metrics to a family of signed evaluation metrics.
We first redefine the relevance indicator into a three-valued signed relevance $\mathrm{rel}_s(i)$:

\begin{equation}
\mathrm{rel}_s(i)=
    \begin{cases}
        +1, & \text{if } i\in\mathcal{P}_u \\
        -\gamma, & \text{if } i\in\mathcal{N}_u \\
        0, & \text{if } i\in\mathcal{U}_u
    \end{cases},
\end{equation}
where $\gamma\geq0$ is a tunable penalty parameter that controls the cost of recommending a negative item relative to the reward of recommending a positive item.
The symmetric default $\gamma=1$ serves as a domain-neutral baseline, treating the retrieval of a negative item as canceling out the reward of a positive item at the same rank.
While this default provides a balanced baseline, $\gamma$ can be adjusted to reflect specific domain characteristics.
For instance, $\gamma>1$ is suitable for domains where negative user experiences carry high cost, such as news feeds or advertising.
Conversely, $\gamma<1$ is appropriate for exploratory settings where the priority lies in content discovery and the penalty for occasional negative feedback is relatively low.
Building upon this definition, we propose a suite of signed evaluation metrics, Signed Recall (SRecall), Signed HR (SHR), and Signed NDCG(SNDCG), formulated as follows:

\textbf{Signed Recall@$K$}
\begin{equation}
    \mathrm{SRecall}@K=\frac{|\mathcal{P}_u\cap\pi_K(u)|-\gamma\cdot|\mathcal{N}_u\cap\pi_K(u)|}{|\mathcal{P}_u|}
\end{equation}

\textbf{Signed HR@$K$}
\begin{equation}
    \mathrm{SHR}@K=\mathbb{I}[|\mathcal{P}_u\cap\pi_K(u)|>0]-\gamma\cdot\mathbb{I}[|\mathcal{N}_u\cap\pi_K(u)|>0]
\end{equation}

\textbf{Signed NDCG@$K$}
\begin{equation}
    \mathrm{SNDCG}@K=\frac{\mathrm{SDCG}@K}{\mathrm{ISDCG}@K}, \quad \mathrm{SDCG}@K=\sum_{k=1}^K\frac{\mathrm{rel}_s(\pi_K(u)[k])}{\log_2(k+1)}, 
\end{equation}
where $\mathrm{ISDCG}@K$ is the ideal signed DCG, computed by placing all positive items in the top positions and all negative items at the bottom of the list.

Each proposed metric rewards the retrieval of liked items while explicitly penalizing the retrieval of disliked items.
Notably, the formulation reduces to conventional metrics at $\gamma=0$, ensuring strict backward compatibility (formalized in Proposition~\ref{prop:characterization}; sensitivity to $\gamma$ analyzed in Appendix~\ref{app:gamma}).
This property allows the proposed metrics to be integrated into existing evaluation protocols as a drop-in replacement, preserving comparability with previously reported results.

%% file: sections/6_experiments.tex
\section{Experiments}
\label{sec:exp}

\subsection{Experimental Setup}
\paragraph{Datasets}
We use five real-world datasets spanning a wide range of negative interaction ratios ($\rho$ from 0.22 to 5.95).
These include three rating-based datasets, Amazon-CDs~\cite{mcauley2013amateurs}, Amazon-Music~\cite{mcauley2013amateurs}, and Epinions~\cite{tang2012etrust}, where ratings$\geq4$ are treated as positive and $<4$ as negative.
We also include two short-video datasets, KuaiRand~\cite{gao2022kuairand} and KuaiRec~\cite{gao2022kuairec}, using click indicators and viewing-duration ratios to define the label.
Detailed statistical characteristics for each dataset are summarized in Table~\ref{tab:datasets}.

\paragraph{Baselines}
We re-evaluate six state-of-the-art sign-aware recommender systems, SiReN~\cite{siren}, SiGRec~\cite{sigrec}, NFARec~\cite{nfarec}, SIGformer~\cite{sigformer}, Pone-GNN~\cite{ponegnn}, and SDCGCL~\cite{sdcgcl} using the proposed metrics.
To provide a comprehensive benchmarking context, we also include four representative unsigned baselines, LightGCN~\cite{lightgcn}, GFormer~\cite{gformer}, LightGCL~\cite{lightgcl}, and XSimGCL~\cite{xsimgcl}.
All models utilize official implementations with recommended hyperparameters.

\begin{table}[h]
\caption{Dataset statistics. $\rho$ denotes the ratio of negative to positive interactions. Ratings $\geq 4$ are positive, $< 4$ are negative for rating-based datasets. ``Both P\&N'' is the fraction of test users with at least one positive and one negative test item.}
\label{tab:datasets}
\centering
\footnotesize
\begin{tabular}{lccccc}
\toprule
\textbf{Dataset} & \textbf{\#Users} & \textbf{\#Items} & \textbf{\#Inter.} & $\rho$ & \textbf{\makecell{Both\\P\&N}} \\
\midrule
Amazon-CDs & 51,267 & 46,464 & 895K & 0.22 & 29.2\% \\
Amazon-Music & 3,472 & 2,498 & 50K & 0.25 & 28.9\% \\
Epinions & 17,894 & 17,660 & 414K & 0.37 & 51.7\% \\
KuaiRand & 16,974 & 4,373 & 263K & 1.25 & 51.6\% \\
KuaiRec & 1,411 & 3,327 & 254K & 5.95 & 72.3\% \\
\bottomrule
\end{tabular}
\end{table}

\paragraph{Evaluation Protocol}
Datasets are split 7:1:2 (train:val:test) with 5-core filtering.
We report standard metrics, Recall@20, HR@20, and NDCG@20, alongside their signed counterparts, SRecall@20, SHR@20, and SNDCG@20, at $\gamma=1$.
All results represent the average of three random seeds where the ranking positions remain stable across all runs.
Since signed metrics reduce to conventional metrics when no negative items are available, we restrict our evaluation to users having at least one positive and one negative item.
This ``Both P\& N'' subpopulation ensures a mathematically rigorous comparison and we compute conventional metrics on the same user group to maintain consistency.

\subsection{Re-evaluation under Signed Metrics}
\label{sec:reeval}

\begin{table*}[p]
\caption{Re-evaluation across all datasets and models ($K=20$, $\gamma=1$). Methods are first split into \textbf{Unsigned} baselines and \textbf{Signed} (sign-aware) methods. Sign-aware methods are further grouped into \textbf{Arch.-driven} (embed sign in model structure with standard or auxiliary-only loss) and \textbf{Loss-driven} (modify the training objective). Within each group, methods are ordered chronologically. Mean\textsubscript{$\pm$std} over three seeds $\{42, 2024, 2025\}$. \textbf{Bold}: best, \underline{underline}: second best per column within each dataset. Significance markers on select SNDCG cells (paired $t$-test on user-level scores, Holm--Bonferroni corrected within each column at $p<0.05$): $\ddagger$ significantly better than all four unsigned baselines (sign-aware methods only); $\dagger$ significantly better than all other nine methods.}
\label{tab:reeval}
\centering
\footnotesize
\setlength{\tabcolsep}{3pt}
\renewcommand{\arraystretch}{0.9}
\begin{adjustbox}{max width=\textwidth, max totalheight=0.85\textheight}
\begin{tabular}{l ll l ccc ccc}
\toprule
& & & & \multicolumn{3}{c}{\textbf{Conventional}} & \multicolumn{3}{c}{\textbf{Signed} ($\gamma=1$)} \\
\cmidrule(lr){5-7} \cmidrule(lr){8-10}
\textbf{Dataset} & \multicolumn{2}{c}{\textbf{Group}} & \textbf{Method} & \textbf{Recall} & \textbf{HR} & \textbf{NDCG} & \textbf{SRecall} & \textbf{SHR} & \textbf{SNDCG} \\
\midrule
\multirow{10}{*}{\rotatebox{90}{Amazon-CDs}}
& \multirow{4}{*}{Unsigned} &
  & LightGCN    & .1351\textsubscript{$\pm$.0003} & .3357\textsubscript{$\pm$.0011} & .0857\textsubscript{$\pm$.0005} & .0338\textsubscript{$\pm$.0011} & .1664\textsubscript{$\pm$.0019} & .0337\textsubscript{$\pm$.0010} \\
& & & GFormer     & .0578\textsubscript{$\pm$.0002} & .1768\textsubscript{$\pm$.0007} & .0370\textsubscript{$\pm$.0002} & .0066\textsubscript{$\pm$.0001} & .0793\textsubscript{$\pm$.0003} & .0099\textsubscript{$\pm$.0002} \\
& & & LightGCL    & \underline{.1383\textsubscript{$\pm$.0004}} & \underline{.3449\textsubscript{$\pm$.0020}} & \underline{.0884\textsubscript{$\pm$.0002}} & .0329\textsubscript{$\pm$.0007} & .1691\textsubscript{$\pm$.0028} & .0339\textsubscript{$\pm$.0001} \\
& & & XSimGCL     & .1373\textsubscript{$\pm$.0006} & .3401\textsubscript{$\pm$.0010} & \textbf{.0886\textsubscript{$\pm$.0001}} & .0328\textsubscript{$\pm$.0009} & .1665\textsubscript{$\pm$.0022} & .0347\textsubscript{$\pm$.0005} \\
\cmidrule(lr){2-10}
& \multirow{6}{*}{Signed} & \multirow{3}{*}{Arch.}
  & SiGRec      & .0821\textsubscript{$\pm$.0037} & .2441\textsubscript{$\pm$.0067} & .0506\textsubscript{$\pm$.0024} & .0201\textsubscript{$\pm$.0003} & .1288\textsubscript{$\pm$.0036} & .0203\textsubscript{$\pm$.0011} \\
& & & NFARec      & \textbf{.1430\textsubscript{$\pm$.0010}} & \textbf{.3464\textsubscript{$\pm$.0025}} & .0840\textsubscript{$\pm$.0008} & .0415\textsubscript{$\pm$.0003} & .1802\textsubscript{$\pm$.0013} & .0363\textsubscript{$\pm$.0004} \\
& & & SIGformer   & .1348\textsubscript{$\pm$.0029} & .3396\textsubscript{$\pm$.0035} & .0865\textsubscript{$\pm$.0019} & .0363\textsubscript{$\pm$.0007} & .1750\textsubscript{$\pm$.0011} & .0360\textsubscript{$\pm$.0004} \\
\cmidrule(lr){3-10}
& & \multirow{3}{*}{Loss}
  & SiReN       & .1259\textsubscript{$\pm$.0036} & .3217\textsubscript{$\pm$.0058} & .0810\textsubscript{$\pm$.0028} & \textbf{.0658\textsubscript{$\pm$.0027}} & \textbf{.2193\textsubscript{$\pm$.0044}} & \textbf{.0526\textsubscript{$\pm$.0021}}$^{\ddagger\dagger}$ \\
& & & Pone-GNN    & .1042\textsubscript{$\pm$.0031} & .2855\textsubscript{$\pm$.0070} & .0641\textsubscript{$\pm$.0047} & \underline{.0585\textsubscript{$\pm$.0016}} & \underline{.2027\textsubscript{$\pm$.0042}} & \underline{.0429\textsubscript{$\pm$.0032}}$^\ddagger$ \\
& & & SDCGCL      & .0869\textsubscript{$\pm$.0299} & .2283\textsubscript{$\pm$.0728} & .0550\textsubscript{$\pm$.0218} & .0426\textsubscript{$\pm$.0119} & .1494\textsubscript{$\pm$.0445} & .0341\textsubscript{$\pm$.0125} \\
\midrule
\multirow{10}{*}{\rotatebox{90}{Amazon-Music}}
& \multirow{4}{*}{Unsigned} &
  & LightGCN    & .2870\textsubscript{$\pm$.0063} & .5293\textsubscript{$\pm$.0147} & .1767\textsubscript{$\pm$.0026} & .0641\textsubscript{$\pm$.0042} & .2340\textsubscript{$\pm$.0153} & .0671\textsubscript{$\pm$.0064} \\
& & & GFormer     & \textbf{.2907\textsubscript{$\pm$.0101}} & \textbf{.5359\textsubscript{$\pm$.0105}} & \textbf{.1860\textsubscript{$\pm$.0072}} & .0947\textsubscript{$\pm$.0235} & .2709\textsubscript{$\pm$.0273} & .0905\textsubscript{$\pm$.0134} \\
& & & LightGCL    & .2699\textsubscript{$\pm$.0083} & .5049\textsubscript{$\pm$.0037} & .1677\textsubscript{$\pm$.0029} & .0730\textsubscript{$\pm$.0089} & .2301\textsubscript{$\pm$.0113} & .0648\textsubscript{$\pm$.0041} \\
& & & XSimGCL     & \underline{.2878\textsubscript{$\pm$.0033}} & \underline{.5312\textsubscript{$\pm$.0034}} & \underline{.1789\textsubscript{$\pm$.0017}} & .0686\textsubscript{$\pm$.0007} & .2317\textsubscript{$\pm$.0047} & .0652\textsubscript{$\pm$.0011} \\
\cmidrule(lr){2-10}
& \multirow{6}{*}{Signed} & \multirow{3}{*}{Arch.}
  & SiGRec      & .1687\textsubscript{$\pm$.0552} & .3702\textsubscript{$\pm$.1090} & .0975\textsubscript{$\pm$.0367} & .0593\textsubscript{$\pm$.0112} & .2010\textsubscript{$\pm$.0251} & .0441\textsubscript{$\pm$.0070} \\
& & & NFARec      & .1861\textsubscript{$\pm$.0048} & .3954\textsubscript{$\pm$.0071} & .1026\textsubscript{$\pm$.0013} & .0343\textsubscript{$\pm$.0117} & .1731\textsubscript{$\pm$.0196} & .0292\textsubscript{$\pm$.0041} \\
& & & SIGformer   & .2812\textsubscript{$\pm$.0072} & .5243\textsubscript{$\pm$.0041} & .1766\textsubscript{$\pm$.0069} & .0664\textsubscript{$\pm$.0057} & .2321\textsubscript{$\pm$.0089} & .0699\textsubscript{$\pm$.0062} \\
\cmidrule(lr){3-10}
& & \multirow{3}{*}{Loss}
  & SiReN       & .2686\textsubscript{$\pm$.0110} & .5064\textsubscript{$\pm$.0111} & .1683\textsubscript{$\pm$.0041} & .1038\textsubscript{$\pm$.0070} & .2763\textsubscript{$\pm$.0126} & .0953\textsubscript{$\pm$.0014} \\
& & & Pone-GNN    & .2534\textsubscript{$\pm$.0045} & .5014\textsubscript{$\pm$.0036} & .1566\textsubscript{$\pm$.0055} & \textbf{.1250\textsubscript{$\pm$.0008}} & \textbf{.3166\textsubscript{$\pm$.0040}} & \underline{.0990\textsubscript{$\pm$.0037}} \\
& & & SDCGCL      & .2541\textsubscript{$\pm$.0056} & .4889\textsubscript{$\pm$.0012} & .1642\textsubscript{$\pm$.0007} & \underline{.1146\textsubscript{$\pm$.0050}} & \underline{.3007\textsubscript{$\pm$.0052}} & \textbf{.0993\textsubscript{$\pm$.0023}} \\
\midrule
\multirow{10}{*}{\rotatebox{90}{Epinions}}
& \multirow{4}{*}{Unsigned} &
  & LightGCN    & .0856\textsubscript{$\pm$.0014} & .2885\textsubscript{$\pm$.0078} & .0566\textsubscript{$\pm$.0004} & .0122\textsubscript{$\pm$.0008} & .1135\textsubscript{$\pm$.0034} & .0147\textsubscript{$\pm$.0012} \\
& & & GFormer     & .0778\textsubscript{$\pm$.0006} & .2704\textsubscript{$\pm$.0011} & .0499\textsubscript{$\pm$.0004} & .0040\textsubscript{$\pm$.0006} & .0868\textsubscript{$\pm$.0015} & .0062\textsubscript{$\pm$.0005} \\
& & & LightGCL    & .0897\textsubscript{$\pm$.0003} & .2966\textsubscript{$\pm$.0031} & .0610\textsubscript{$\pm$.0007} & .0155\textsubscript{$\pm$.0021} & .1259\textsubscript{$\pm$.0055} & .0190\textsubscript{$\pm$.0006} \\
& & & XSimGCL     & \textbf{.0972\textsubscript{$\pm$.0006}} & \textbf{.3138\textsubscript{$\pm$.0004}} & \textbf{.0658\textsubscript{$\pm$.0005}} & .0134\textsubscript{$\pm$.0020} & .1236\textsubscript{$\pm$.0028} & .0178\textsubscript{$\pm$.0006} \\
\cmidrule(lr){2-10}
& \multirow{6}{*}{Signed} & \multirow{3}{*}{Arch.}
  & SiGRec      & .0633\textsubscript{$\pm$.0009} & .2370\textsubscript{$\pm$.0017} & .0410\textsubscript{$\pm$.0004} & .0040\textsubscript{$\pm$.0013} & .0875\textsubscript{$\pm$.0026} & .0086\textsubscript{$\pm$.0012} \\
& & & NFARec      & .0853\textsubscript{$\pm$.0022} & .2848\textsubscript{$\pm$.0028} & .0501\textsubscript{$\pm$.0004} & .0012\textsubscript{$\pm$.0007} & .0830\textsubscript{$\pm$.0033} & .0053\textsubscript{$\pm$.0009} \\
& & & SIGformer   & \underline{.0935\textsubscript{$\pm$.0019}} & \underline{.3088\textsubscript{$\pm$.0031}} & \underline{.0615\textsubscript{$\pm$.0023}} & .0102\textsubscript{$\pm$.0018} & .1143\textsubscript{$\pm$.0094} & .0146\textsubscript{$\pm$.0035} \\
\cmidrule(lr){3-10}
& & \multirow{3}{*}{Loss}
  & SiReN       & .0755\textsubscript{$\pm$.0013} & .2556\textsubscript{$\pm$.0011} & .0491\textsubscript{$\pm$.0009} & .0357\textsubscript{$\pm$.0026} & .1637\textsubscript{$\pm$.0033} & .0291\textsubscript{$\pm$.0010} \\
& & & Pone-GNN    & .0809\textsubscript{$\pm$.0006} & .2735\textsubscript{$\pm$.0019} & .0540\textsubscript{$\pm$.0005} & \underline{.0393\textsubscript{$\pm$.0012}} & \underline{.1778\textsubscript{$\pm$.0029}} & \underline{.0337\textsubscript{$\pm$.0002}} \\
& & & SDCGCL      & .0787\textsubscript{$\pm$.0084} & .2641\textsubscript{$\pm$.0264} & .0534\textsubscript{$\pm$.0058} & \textbf{.0433\textsubscript{$\pm$.0039}} & \textbf{.1806\textsubscript{$\pm$.0174}} & \textbf{.0347\textsubscript{$\pm$.0037}} \\
\midrule
\multirow{10}{*}{\rotatebox{90}{KuaiRand}}
& \multirow{4}{*}{Unsigned} &
  & LightGCN    & .1200\textsubscript{$\pm$.0166} & .2202\textsubscript{$\pm$.0236} & .0586\textsubscript{$\pm$.0080} & $-$.0341\textsubscript{$\pm$.0009} & .0041\textsubscript{$\pm$.0036} & $-$.0101\textsubscript{$\pm$.0013} \\
& & & GFormer     & \textbf{.1442\textsubscript{$\pm$.0006}} & \textbf{.2563\textsubscript{$\pm$.0013}} & \textbf{.0718\textsubscript{$\pm$.0003}}$^\dagger$ & $-$.0352\textsubscript{$\pm$.0025} & .0111\textsubscript{$\pm$.0017} & $-$.0082\textsubscript{$\pm$.0010} \\
& & & LightGCL    & .1033\textsubscript{$\pm$.0115} & .1916\textsubscript{$\pm$.0167} & .0507\textsubscript{$\pm$.0069} & $-$.0322\textsubscript{$\pm$.0045} & .0004\textsubscript{$\pm$.0006} & $-$.0112\textsubscript{$\pm$.0028} \\
& & & XSimGCL     & .1346\textsubscript{$\pm$.0013} & .2414\textsubscript{$\pm$.0038} & .0662\textsubscript{$\pm$.0008} & $-$.0371\textsubscript{$\pm$.0012} & .0089\textsubscript{$\pm$.0012} & $-$.0115\textsubscript{$\pm$.0009} \\
\cmidrule(lr){2-10}
& \multirow{6}{*}{Signed} & \multirow{3}{*}{Arch.}
  & SiGRec      & .1076\textsubscript{$\pm$.0035} & .1995\textsubscript{$\pm$.0051} & .0495\textsubscript{$\pm$.0024} & $-$.0371\textsubscript{$\pm$.0035} & $-$.0042\textsubscript{$\pm$.0048} & $-$.0143\textsubscript{$\pm$.0016} \\
& & & NFARec      & .1013\textsubscript{$\pm$.0002} & .1896\textsubscript{$\pm$.0012} & .0463\textsubscript{$\pm$.0006} & $-$.0365\textsubscript{$\pm$.0030} & $-$.0028\textsubscript{$\pm$.0013} & $-$.0157\textsubscript{$\pm$.0005} \\
& & & SIGformer   & \underline{.1393\textsubscript{$\pm$.0016}} & \underline{.2483\textsubscript{$\pm$.0039}} & \underline{.0680\textsubscript{$\pm$.0014}} & $-$.0331\textsubscript{$\pm$.0013} & .0128\textsubscript{$\pm$.0008} & $-$.0075\textsubscript{$\pm$.0011} \\
\cmidrule(lr){3-10}
& & \multirow{3}{*}{Loss}
  & SiReN       & .0785\textsubscript{$\pm$.0010} & .1514\textsubscript{$\pm$.0014} & .0374\textsubscript{$\pm$.0007} & \textbf{.0027\textsubscript{$\pm$.0026}} & \textbf{.0400\textsubscript{$\pm$.0046}} & \underline{.0039\textsubscript{$\pm$.0010}}$^\ddagger$ \\
& & & Pone-GNN    & .1045\textsubscript{$\pm$.0045} & .1980\textsubscript{$\pm$.0072} & .0516\textsubscript{$\pm$.0025} & $-$.0054\textsubscript{$\pm$.0007} & \underline{.0393\textsubscript{$\pm$.0017}} & \textbf{.0045\textsubscript{$\pm$.0009}}$^\ddagger$ \\
& & & SDCGCL      & .0832\textsubscript{$\pm$.0094} & .1591\textsubscript{$\pm$.0172} & .0428\textsubscript{$\pm$.0041} & \underline{$-$.0028\textsubscript{$\pm$.0035}} & .0306\textsubscript{$\pm$.0084} & \underline{.0039\textsubscript{$\pm$.0018}}$^\ddagger$ \\
\midrule
\multirow{10}{*}{\rotatebox{90}{KuaiRec}}
& \multirow{4}{*}{Unsigned} &
  & LightGCN    & .0026\textsubscript{$\pm$.0012} & .0134\textsubscript{$\pm$.0147} & .0015\textsubscript{$\pm$.0012} & $-$4.177\textsubscript{$\pm$.3370} & $-$.9918\textsubscript{$\pm$.0142} & $-$1.828\textsubscript{$\pm$.1845} \\
& & & GFormer     & .0019\textsubscript{$\pm$.0000} & .0046\textsubscript{$\pm$.0006} & .0007\textsubscript{$\pm$.0001} & $-$4.372\textsubscript{$\pm$.0028} & $-$1.000 & $-$1.931\textsubscript{$\pm$.0067} \\
& & & LightGCL    & .0017\textsubscript{$\pm$.0002} & .0046\textsubscript{$\pm$.0006} & .0009\textsubscript{$\pm$.0003} & $-$4.057\textsubscript{$\pm$.0169} & $-$.9990\textsubscript{$\pm$.0010} & $-$1.781\textsubscript{$\pm$.0064} \\
& & & XSimGCL     & .0134\textsubscript{$\pm$.0030} & .0683\textsubscript{$\pm$.0032} & .0083\textsubscript{$\pm$.0006} & $-$.8810\textsubscript{$\pm$.1280} & $-$.4608\textsubscript{$\pm$.0148} & $-$.3647\textsubscript{$\pm$.0500} \\
\cmidrule(lr){2-10}
& \multirow{6}{*}{Signed} & \multirow{3}{*}{Arch.}
  & SiGRec      & \underline{.0669\textsubscript{$\pm$.0018}} & \underline{.2428\textsubscript{$\pm$.0023}} & \underline{.0412\textsubscript{$\pm$.0011}}$^\ddagger$ & $-$.1615\textsubscript{$\pm$.0151} & .1497\textsubscript{$\pm$.0074} & $-$.0569\textsubscript{$\pm$.0053}$^\ddagger$ \\
& & & NFARec      & .0019\textsubscript{$\pm$.0000} & .0049\textsubscript{$\pm$.0010} & .0007\textsubscript{$\pm$.0001} & $-$4.221\textsubscript{$\pm$.0887} & $-$.9977\textsubscript{$\pm$.0032} & $-$1.622\textsubscript{$\pm$.1021} \\
& & & SIGformer$^\ast$ & \textbf{.0830\textsubscript{$\pm$.0052}} & \textbf{.2676\textsubscript{$\pm$.0069}} & \textbf{.0496\textsubscript{$\pm$.0014}}$^{\ddagger\dagger}$ & \textbf{.0632\textsubscript{$\pm$.0048}} & \textbf{.2314\textsubscript{$\pm$.0055}} & \textbf{.0418\textsubscript{$\pm$.0015}}$^{\ddagger\dagger}$ \\
\cmidrule(lr){3-10}
& & \multirow{3}{*}{Loss}
  & SiReN       & .0249\textsubscript{$\pm$.0347} & .0899\textsubscript{$\pm$.1269} & .0153\textsubscript{$\pm$.0225}$^\ddagger$ & $-$3.303\textsubscript{$\pm$1.369} & $-$.7827\textsubscript{$\pm$.3764} & $-$1.345\textsubscript{$\pm$.6819} \\
& & & Pone-GNN    & .0515\textsubscript{$\pm$.0162} & .1974\textsubscript{$\pm$.0308} & .0314\textsubscript{$\pm$.0062}$^\ddagger$ & $-$2.146\textsubscript{$\pm$1.094} & $-$.5915\textsubscript{$\pm$.1480} & $-$.6516\textsubscript{$\pm$.3664} \\
& & & SDCGCL      & .0614\textsubscript{$\pm$.0019} & .2255\textsubscript{$\pm$.0010} & .0382\textsubscript{$\pm$.0008}$^\ddagger$ & \underline{.0270\textsubscript{$\pm$.0050}} & \underline{.1572\textsubscript{$\pm$.0099}} & \underline{.0243\textsubscript{$\pm$.0037}}$^\ddagger$ \\
\bottomrule
\multicolumn{10}{l}{\footnotesize $^\ast$ SIGformer uses $\beta = -0.2$ on KuaiRec, which reverses the negative feedback direction.}
\end{tabular}
\end{adjustbox}
\end{table*}

We now systematically re-evaluate the state-of-the-art models using our proposed signed metrics.
Table \ref{tab:reeval} presents the comprehensive evaluation results across all five benchmark datasets.

First, the top-ranked method fundamentally differs between conventional and signed metrics.
On Amazon-CDs, the unsigned baseline XSimGCL (.0886) holds the top NDCG position but under SNDCG the top position shifts to the sign-aware method, SiReN (.0526), with other \emph{Loss-driven} sign-aware methods re-ordering accordingly to secure the top ranks.
This pattern repeats on KuaiRand, where the NDCG-best unsigned method GFormer (.0718) achieves a negative SNDCG score ($-$.0082), indicating that it actively ranks disliked items above liked items.
Meanwhile, Pone-GNN, another \emph{Loss-driven} method, captures the top positive SNDCG position.
Notably, on KuaiRec, where negative interactions dominates, the alignment between relevance and valence objectives shifts, and Architecture-driven methods achieve the highest SNDCG, while Loss-driven methods collapse to negative scores.

Second, conventional metrics consistently fail to identify models that successfully leverage negative feedback.
On datasets with high-$\rho$, unsigned methods with near zero NDCG exhibit negative SNDCG scores such as GFormer ($-$1.931).
This reveals that their poor conventional performance actually masks a catastrophic failure rather than mild inaccuracy.
Conversely, on low-to-mid-$\rho$ datasets, sign-aware methods with better valence separation dominate the rankings under SNDCG, with Loss-driven methods consistently securing the top positions.
This confirms that negative feedback training provides significant utility that remains largely invisible under conventional relevance-only evaluation.

Finally, we verify that the observed ranking shifts are statistically robust and not mere noise.
As indicated by the markers in Table \ref{tab:reeval}, we perform paired t-tests on user-level scores with Holm-Bonferroni correction ($\alpha=0.05$).
The top-performing sign-aware methods consistently carry $\ddagger$ markers, indicating that they significantly outperform every unsigned baseline under SNDCG.
Notably, these methods also achieve $\dagger$ markers, outperforming all nine competing models.
The significance structure aligns with the NDCG-to-SNDCG re-ordering across all datasets, whereas unsigned methods often lose their statistical dominance under our metrics.
Because our comparisons are paired at the user level, these conclusions remain highly robust against between user heterogeneity, validating the critical need for our proposed evaluation metrics.

\subsection{Proof of Concept}
To verify the practical utility of our proposed metrics, we apply a minimal auxiliary loss, $\mathcal{L}_\mathrm{sign}$, to one baseline as a controlled test rather than a method proposal.
Specifically, $\mathcal{L}_\mathrm{sign}$ trains a linear classifier on the element-wise embedding product $\mathbf{e}_u\odot\mathbf{e}_i$ (which yields a gradient component orthogonal to $\mathbf{e}_u$; Proposition~\ref{prop:sign_decoupling}) to predict whether an interaction is positive or negative:
\begin{equation}
    \mathcal{L}_\mathrm{sign}=
    -\frac{1}{|\mathcal{B}|}\sum_{(u,i,y)\in\mathcal{B}}\left[y\log\hat{p}_{ui}+(1-y)\log(1-\hat{p}_{ui})\right],
\end{equation}
where $y\in\{0,1\}$ indicates positive $(y=1)$ or negative $(y=0)$ feedback, $\hat{p}_{ui}=\sigma(\mathbf{w}^\top(\mathbf{e}_u\odot\mathbf{e}_i)+b)$, and $\mathbf{w}\in\mathbb{R}^d, b\in\mathbb{R}$ are learnable.
We apply this to the SIGformer with the combined objective $\mathcal{L}=\mathcal{L}_\mathrm{base}+\lambda_\mathrm{sign}\cdot\mathcal{L}_\mathrm{sign}$, where $\mathcal{L}_\mathrm{base}$ is SIGformer's signed BPR loss.
At inference, the classifier output is additionally incorporated into the ranking score as $\hat{y}_{ui}^{\text{final}}=\hat{y}_{ui}+\lambda_{\text{inf}}\cdot(2\hat{p}_{ui}-1)$, where $\hat{y}_{ui}=\mathbf{e}_u^\top\mathbf{e}_i$ is the original relevance score, $(2\hat{p}_{ui}-1)\in[-1, +1]$ acts as a valence adjustment scaled by $\lambda_\mathrm{inf}\geq0$, and $\lambda_{\text{inf}}=0$ recovers the original model.
We search over $\lambda_\mathrm{sign}$ and $\lambda_\mathrm{inf}$ and select per-dataset the configuration that maximizes SNDCG while keeping NDCG within 5\% of the baseline (details in Appendix~\ref{app:impl_lsign}).

As shown in Table \ref{tab:lsign}, $\mathcal{L}_\mathrm{sign}$ improves signed metrics on four of five datasets while preserving conventional relevance.
SNDCG rises by 26\% on Amazon-CDs, 46\% on Amazon-Music, and nearly triples on Epinions (+187\%).
The baseline SNDCG is negative ($-.0070$) on KuaiRand, meaning the model actively ranks disliked items above liked ones for the average user.
$\mathcal{L}_\mathrm{sign}$ reverses this to a positive value at only 3.9\% relevance cost.
Across all five datasets, $\Delta\text{Pr.}$ is consistently positive, confirming that $\mathcal{L}_\mathrm{sign}$ successfully injects valence information into the embedding space.
Together, these results confirm that the proposed signed metrics provide actionable training signals and that models can be explicitly guided toward valence-aware behavior without sacrificing conventional retrieval performance.

KuaiRec is the exception, as the baseline already achieves strong valence separation and no configuration satisfies the 5\% NDCG constraint while improving SNDCG.
This outcome is informative rather than problematic.
It demonstrates that signed metrics do not blindly reward valence separation, but jointly evaluate relevance and valence, penalizing configurations that sacrifice relevance for valence.

\begin{table}[h!]
\caption{Effect of $\mathcal{L}_\mathrm{sign}$ applied to SIGformer ($K=20$, $\gamma=1$). Per-dataset configurations selected for best SNDCG while keeping NDCG within 5\% of the baseline; values in parentheses denote ($\lambda_\text{sign}$, $\lambda_\text{inf}$). Evaluated on users with both positive and negative test items. \textbf{Bold}: better of the two rows per metric.}
\label{tab:lsign}
\footnotesize
\centering
\renewcommand{\arraystretch}{0.90}
\begin{tabular}{ll ccccc}
\toprule
\textbf{Dataset} & \textbf{Model} & \textbf{NDCG} & \textbf{SRecall} & \textbf{SHR} & \textbf{SNDCG} & $\Delta$\textbf{Pr.} \\
\midrule
\multirow{2}{*}{Amazon-CDs}
& SIGformer                         & .0842 & .0359 & .1507 & .0355 & -- \\
& SIGformer $+\mathcal{L}_\text{sign}$ (.01, 2) & .0829 & \textbf{.0520} & \textbf{.1798} & \textbf{.0447} & +.067 \\
\midrule
\multirow{2}{*}{Amazon-Music}
& SIGformer                         & .1674 & .0595 & .1676 & .0615 & -- \\
& SIGformer $+\mathcal{L}_\text{sign}$ (.05, 2) & .1653 & \textbf{.1036} & \textbf{.2258} & \textbf{.0900} & +.052 \\
\midrule
\multirow{2}{*}{Epinions}
& SIGformer                         & .0587 & .0081 & .0919 & .0104 & -- \\
& SIGformer $+\mathcal{L}_\text{sign}$ (.1, 1)  & .0595 & \textbf{.0326} & \textbf{.1602} & \textbf{.0299} & +.131 \\
\midrule
\multirow{2}{*}{KuaiRand}
& SIGformer                         & .0684 & $-$.0319 & .0141 & $-$.0070 & -- \\
& SIGformer $+\mathcal{L}_\text{sign}$ (.1, 2)  & .0657 & \textbf{$-$.0035} & \textbf{.0519} & \textbf{.0071} & +.043 \\
\midrule
\multirow{2}{*}{KuaiRec}
& SIGformer$^*$                     & .0481 & \textbf{.0582} & \textbf{.2255} & \textbf{.0401} & -- \\
& SIGformer $+\mathcal{L}_\text{sign}$ (.01, 5) & .0338 & .0444 & .1775 & .0271 & +.038 \\
\bottomrule
\multicolumn{7}{l}{\footnotesize $^*$ On KuaiRec, no $+\mathcal{L}_\text{sign}$ config meets the 5\% NDCG constraint.}
\end{tabular}%
\end{table}

%% file: sections/7_conclusion.tex
\section{Conclusion and Limitations}
\label{sec:conclusion}

In this work, we identify that sign-aware recommender systems paradoxically fail to separate positive from negative items, and re-evaluate existing methods using signed metrics that assign negative utility to disliked items in the top-$K$ list.
While these metrics provide the societal benefit of reducing exposure to disliked content, potential over-optimization may suppress content exploration.

However, our analysis focuses on graph-based models using inner-product scoring leaves generalization to sequential or LLM-based recommender systems for future work.
Furthermore, our proof-of-concept auxiliary loss ($\mathcal{L}_\mathrm{sign}$) reveals that enforcing valence separation in dense negative environments can conflict with performance preservation.
Ultimately, operating on the principle that \emph{what gets measured gets managed}, we hope this work encourages the community to evaluate recommender systems not only for what they retrieve, but also for what they must avoid.

%% file: sections/8_appendix.tex
\appendix

\section{The Evaluation Gap}
\label{app:related}

This appendix expands on the position of signed metrics within the broader evaluation-methodology literature referenced in \S\ref{sec:related}. Evaluation methodology has been extensively studied~\cite{herlocker2004evaluating, gunawardana2009survey, bauer2024landscape, mcnee2006being}, spurring research into beyond-accuracy metrics~\cite{ge2010beyond, vargas2011rank, kaminskas2016diversity} and critical examinations of protocol biases~\cite{krichene2020sampled, rendle2019evaluation, zhao2022revisiting, steck2013evaluation}. Standard ranking metrics derive from a graded-relevance framework~\cite{jarvelin2002cumulated} that, with binary labels, assigns identical scores to negative and unobserved items. None of these advances address the treatment of \emph{explicit} negatives: a design predating sign-aware training that has not kept pace with it.

\textbf{Position relative to IR metric tradition.} Information retrieval has a long tradition of metrics beyond binary relevance: graded-relevance NDCG~\cite{jarvelin2002cumulated, kekalainen2005binary}, the cascade-based ERR~\cite{chapelle2009err}, and user-model metrics such as RBP~\cite{moffat2008rbp}. These all assume non-negative relevance grades and cannot distinguish rejected items from unseen ones. Our signed metrics differ in two ways: they admit \emph{negative} relevance values encoding explicit rejection (not merely lower positive grades), which requires a modified IDCG that places positives first and negatives last so that the metric remains bounded as $|\mathcal{N}_u|$ grows; and the $\gamma\!=\!0$ limit recovers standard binary-relevance NDCG (Proposition~\ref{prop:characterization}(i)), enabling drop-in replacement rather than a separate evaluation track. Signed metrics are thus a specialization of the graded-relevance tradition to the signed-feedback setting, not a departure.

Several neighboring evaluation traditions might appear to address this gap but each captures an orthogonal property: graded relevance~\cite{jarvelin2002cumulated} encodes \emph{degrees} of preference rather than the liked/disliked distinction; beyond-accuracy metrics (diversity, serendipity~\cite{ge2010beyond, kaminskas2016diversity}) evaluate set-level properties; calibrated recommendations~\cite{steck2018calibrated} match category proportions without enforcing within-list ordering; fairness-of-exposure~\cite{singh2018fairness} partitions items by protected attributes, not user-revealed valence; and debiased evaluation~\cite{yang2018unbiased, saito2020unbiased} corrects which items are observed but leaves how they are scored unchanged. None operationalize the qualitative positive/negative distinction at the list level.

\section{Theoretical Properties of Signed Metrics}
\label{app:theory}

This appendix establishes two formal results for the signed metrics introduced in \S\ref{sec:signed_metrics}. Proposition~\ref{prop:characterization} (\S\ref{app:theory_char}) characterizes signed metrics in relation to their conventional counterparts; this provides the formal grounding for the $\gamma$-sensitivity analysis in Appendix~\ref{app:gamma} and explains the ranking changes observed in the re-evaluation of \S\ref{sec:reeval}. Corollary~\ref{cor:monotonicity} (\S\ref{app:theory_mono}) further establishes that the dependence of any signed metric on $\gamma$ is affine, licensing $\gamma$ as an evaluation-time \emph{dial} rather than a fixed hyperparameter.

For convenience, we recall the binary relevance indicator implicit in the conventional metrics of \S\ref{sec:prelim}:
\begin{equation}
\label{eq:binary_rel}
\mathrm{rel}_u(i) = \mathbb{I}[i \in \mathcal{P}_u],
\end{equation}
which assigns identical scores to negative and unobserved items, motivating the signed extension $\mathrm{rel}_s$ of \S\ref{sec:signed_metrics}.

\textbf{A note on notation.} Throughout this appendix, $\sigma$ denotes a permutation of a ranking. The same symbol $\sigma(\cdot)$ is reused in Appendix~\ref{app:gradient} for the sigmoid function, as is standard in that context; the meaning is unambiguous from context.

\subsection{Characterization}
\label{app:theory_char}

\begin{proposition}[Characterization of Signed Metrics]
\label{prop:characterization}
Let $\pi$ be a ranking of the candidate items, evaluated at cutoff $K$, and let $\sigma(\pi)$ denote any ranking obtained by permuting negative and unobserved items in $\pi$ while keeping positive items fixed. For each pair of conventional and signed metrics $(M, M_s) \in \{(\text{NDCG}, \text{SNDCG}), (\text{Recall}, \text{SRecall}), (\text{HR}, \text{SHR})\}$ at cutoff $K$:

\emph{(i) Backward compatibility.} $M_s = M$ when $\gamma = 0$.

\emph{(ii) Valence blindness of $M$.} $M(\pi) = M(\sigma(\pi))$.

\emph{(iii) Valence sensitivity of $M_s$.} For $\gamma > 0$, suppose $\sigma$ is the transposition that swaps a negative item $n$ at rank $k \leq K$ of $\pi$ with an unobserved item at rank $k' > k$. Then $M_s(\sigma(\pi)) > M_s(\pi)$ holds: for SNDCG, with no further condition; for SRecall, when $k' > K$ (the negative leaves the top-$K$); for SHR, when $k' > K$ and $n$ is the unique negative item in $\mathcal{N}_u \cap \pi_K$.
\end{proposition}

Parts~(ii) and~(iii) together establish that signed metrics resolve the exact limitation of conventional metrics. Part~(iii) is essentially an \emph{adversarial-improvement test}: any swap that demotes a negative item below an unobserved one strictly improves SNDCG, and improves SRecall and SHR under their respective metric-specific conditions, even though the conventional score is unchanged by the same swap (part~(ii)). The metric-specific conditions reflect the coarseness of set-based metrics: SRecall and SHR cannot register within-top-$K$ position changes, and SHR additionally cannot register changes when other negatives remain in the top-$K$. Part~(i) ensures that adopting signed metrics is cost-free: existing pipelines recover conventional values by setting $\gamma = 0$.

\textbf{Concrete illustration.} Consider $K\!=\!4$, $\mathcal{P}_u\!=\!\{A\}$, $\mathcal{N}_u\!=\!\{B\}$, and unobserved items $C, D, E$. Compare two top-$4$ rankings $\pi = (A, B, C, D)$ and $\sigma(\pi) = (A, D, C, B)$: both contain the same items in the top-$4$ with positive $A$ at rank~$1$, and $\sigma$ is the transposition that swaps negative $B$ (at rank~$2$ of $\pi$) with unobserved $D$ (at rank~$4$ of $\pi$). Conventional NDCG@$4$, Recall@$4$, and HR@$4$ assign identical scores to $\pi$ and $\sigma(\pi)$ (part~(ii)). SNDCG@$4$, however, strictly prefers $\sigma(\pi)$ since the negative penalty $-\gamma$ is weighted at $1/\log_2(5)$ instead of the larger $1/\log_2(3)$ (part~(iii)). SRecall@$4$ and SHR@$4$ remain unchanged in this example because the swap is within the top-$4$ ($k'=K=4$), consistent with the metric-specific conditions in part~(iii).

\begin{proof}[Proof of Proposition~\ref{prop:characterization}]
We prove each part in turn.

\emph{(i) Backward compatibility.} Setting $\gamma = 0$ in the signed relevance $\mathrm{rel}_s$ defined in \S\ref{sec:signed_metrics} yields $\textit{rel}_s(i) = 1$ if $i \in \mathcal{P}_u$ and $\textit{rel}_s(i) = 0$ otherwise, identical to the binary relevance Eq.~\ref{eq:binary_rel}. All subsequent computations are therefore identical, so $M_s = M$ for each pair.

\emph{(ii) Valence blindness.} Under binary relevance, $\text{rel}_u(i) = 0$ for all $i \in \mathcal{N}_u \cup \mathcal{U}_u$. For \textbf{NDCG@}$K$: permuting such items does not change DCG@$K$ since their contributions are identically zero; IDCG@$K$ depends only on $|\mathcal{P}_u|$ and $K$, unchanged by the permutation. For \textbf{Recall@}$K$ and \textbf{HR@}$K$: both depend only on $|\mathcal{P}_u \cap \pi_K|$, which is preserved since positive items are held fixed.

\emph{(iii) Valence sensitivity.} Let $n \in \mathcal{N}_u$ be the negative item moved from rank $k \leq K$ to rank $k' > k$, and let $j \in \mathcal{U}_u$ be the unobserved item that takes its place at rank $k$. We argue separately for each metric.

\textbf{SNDCG@$K$ with $k' \leq K$.} The SDCG changes by
\begin{equation}
\Delta = \frac{\textit{rel}_s(j) - \textit{rel}_s(n)}{\log_2(k+1)} + \frac{\textit{rel}_s(n) - \textit{rel}_s(j)}{\log_2(k'+1)}.
\end{equation}
Since $\textit{rel}_s(n) = -\gamma < 0 \leq \textit{rel}_s(j)$, we have $\textit{rel}_s(j) - \textit{rel}_s(n) > 0$. Combined with $\tfrac{1}{\log_2(k+1)} > \tfrac{1}{\log_2(k'+1)}$ for $k < k'$, this gives $\Delta > 0$.

\textbf{SNDCG@$K$ with $k' > K$.} The negative is removed from the top-$K$ and replaced with $j$: SDCG changes by $(\textit{rel}_s(j) - \textit{rel}_s(n)) / \log_2(k+1) > 0$.

In both subcases IDCG$_s$@$K$ is unchanged, so SNDCG@$K$ strictly increases.

\textbf{SRecall@$K$ (with $k' > K$).} The transposition removes a negative from the top-$K$ ($|\mathcal{N}_u \cap \pi_K|$ decreases by $1$) without changing $|\mathcal{P}_u \cap \pi_K|$ (since $j$ is unobserved). The numerator $|\mathcal{P}_u \cap \pi_K| - \gamma|\mathcal{N}_u \cap \pi_K|$ thus strictly increases by $\gamma$, while the denominator $|\mathcal{P}_u|$ is unchanged. Hence SRecall@$K$ strictly increases by $\gamma/|\mathcal{P}_u|$.

\textbf{SHR@$K$ (with $k' > K$ and $n$ the unique negative in $\mathcal{N}_u \cap \pi_K$).} Under the uniqueness condition, $\mathcal{N}_u \cap \sigma(\pi)_K = \emptyset$, so $\mathbb{I}[|\mathcal{N}_u \cap \sigma(\pi)_K| > 0] = 0$ while $\mathbb{I}[|\mathcal{N}_u \cap \pi_K| > 0] = 1$. The penalty term $-\gamma\,\mathbb{I}[|\mathcal{N}_u \cap \pi_K| > 0]$ thus increases from $-\gamma$ to $0$. Since $j$ is unobserved, $|\mathcal{P}_u \cap \pi_K|$ is unchanged, and the positive indicator $\mathbb{I}[|\mathcal{P}_u \cap \pi_K| > 0]$ is unchanged. Hence SHR@$K$ strictly increases by exactly $\gamma$. Without the uniqueness condition, the indicator $\mathbb{I}[|\mathcal{N}_u \cap \pi_K| > 0]$ remains $1$ after the swap and SHR@$K$ is unchanged, reflecting its binary penalty structure.
\end{proof}

\subsection{Monotonicity in \texorpdfstring{$\gamma$}{gamma}}
\label{app:theory_mono}

\begin{corollary}[Monotonicity and Dominance in $\gamma$]
\label{cor:monotonicity}
Let $M_s \in \{\mathrm{SNDCG}, \mathrm{SRecall}, \mathrm{SHR}\}$ at cutoff $K$.

\emph{(i) Monotonicity.} For any ranking $\pi$ with at least one negative item in the top-$K$, $M_s(\pi; \gamma)$ is strictly decreasing and affine in $\gamma$.

\emph{(ii) Dominance criterion.} For two rankings $\pi_A, \pi_B$ with $M_s(\pi_A; 0) \geq M_s(\pi_B; 0)$, we have $M_s(\pi_A; \gamma) \geq M_s(\pi_B; \gamma)$ for all $\gamma \geq 0$ if and only if $b(\pi_A) \leq b(\pi_B)$, where $b(\pi)$ is the metric-specific quantity derived in the proof.
\end{corollary}

Part~(i) licenses $\gamma$ as a \emph{decision-theoretic dial}: $\gamma \mapsto M_s(\pi; \gamma)$ is an affine function with non-positive slope $-b(\pi)$ (whose form is given in the proof), so practitioners read off valence-vs-relevance trade-offs by inspecting a single ranking at multiple $\gamma$ (Appendix~\ref{app:gamma}).

Part~(ii) explains the crossovers in our re-evaluation (\S\ref{sec:reeval}). Suppose two methods are \emph{comparable at $\gamma\!=\!0$}, meaning they have similar conventional NDCG (or Recall/HR) on the same data. If their negative-infiltration rates differ, the criterion in~(ii) is violated, so one method's curve must overtake the other's at some $\gamma > 0$. This is exactly the crossing behavior visible in Figure~\ref{fig:gamma}: methods with comparable $\gamma\!=\!0$ scores diverge as $\gamma$ grows, with the slope difference reflecting the gap in their (metric-specific) negative content within the top-$K$.

\begin{proof}[Proof of Corollary~\ref{cor:monotonicity}]
For any ranking $\pi$, let $\mathcal{N}_u^K \!:=\! \mathcal{N}_u \cap \pi_K$ denote the negatives in the top-$K$. We will show that each signed metric admits a decomposition of the form
\begin{equation*}
M_s(\pi; \gamma) = a(\pi) - \gamma \cdot b(\pi),
\end{equation*}
where $a(\pi)$ collects all $\gamma$-independent terms and $b(\pi) \geq 0$ depends only on $\mathcal{N}_u \cap \pi_K$ and takes a metric-specific form (derived below). The corollary then follows immediately.

\emph{(i) Monotonicity.} For SNDCG@$K$, the signed relevance $\mathrm{rel}_s$ (\S\ref{sec:signed_metrics}) splits the numerator as $\text{SDCG@}K = \text{DCG@}K - \gamma \cdot D(\pi)$, where $D(\pi) = \sum_{k:\pi(k) \in \mathcal{N}_u^K} 1/\log_2(k+1) \geq 0$ is the position-weighted negative count. The denominator $\text{IDCG}_s@K$ is independent of $\pi$ (it is determined by the user-level interaction sizes and $K$). Setting $a(\pi) := \text{DCG@}K / \text{IDCG}_s@K$ and $b(\pi) := D(\pi) / \text{IDCG}_s@K$ yields the desired decomposition. SRecall@$K$ and SHR@$K$ decompose analogously with $b(\pi) = |\mathcal{N}_u^K|/|\mathcal{P}_u|$ and $b(\pi) = \mathbf{1}[|\mathcal{N}_u^K| > 0]$ respectively. In all three cases, $b(\pi) > 0$ iff $\mathcal{N}_u^K \neq \emptyset$, giving strict decrease in $\gamma$ under that condition.

\emph{(ii) Dominance criterion.} Since each signed metric is affine in $\gamma$ with slope $-b(\pi) \leq 0$, the difference $M_s(\pi_A; \gamma) - M_s(\pi_B; \gamma) = [a(\pi_A) - a(\pi_B)] - \gamma \cdot [b(\pi_A) - b(\pi_B)]$ is itself affine in $\gamma$. Given $M_s(\pi_A; 0) \geq M_s(\pi_B; 0)$ (i.e., $a(\pi_A) \geq a(\pi_B)$), this difference remains non-negative for all $\gamma \geq 0$ if and only if $b(\pi_A) \leq b(\pi_B)$, which is the stated condition.
\end{proof}

\section{Sensitivity to \texorpdfstring{$\gamma$}{gamma}}
\label{app:gamma}

This appendix analyzes how ranking conclusions depend on the penalty parameter $\gamma$. We first examine $\gamma$ sensitivity across all five benchmarks, document the underlying mechanism via top-$K$ negative infiltration statistics, and discuss how to choose $\gamma$ in practice.

\subsection{\texorpdfstring{$\gamma$}{gamma} Sensitivity Across Datasets}
\label{app:gamma_main}

\begin{figure*}[t]
\centering
\includegraphics[width=\textwidth]{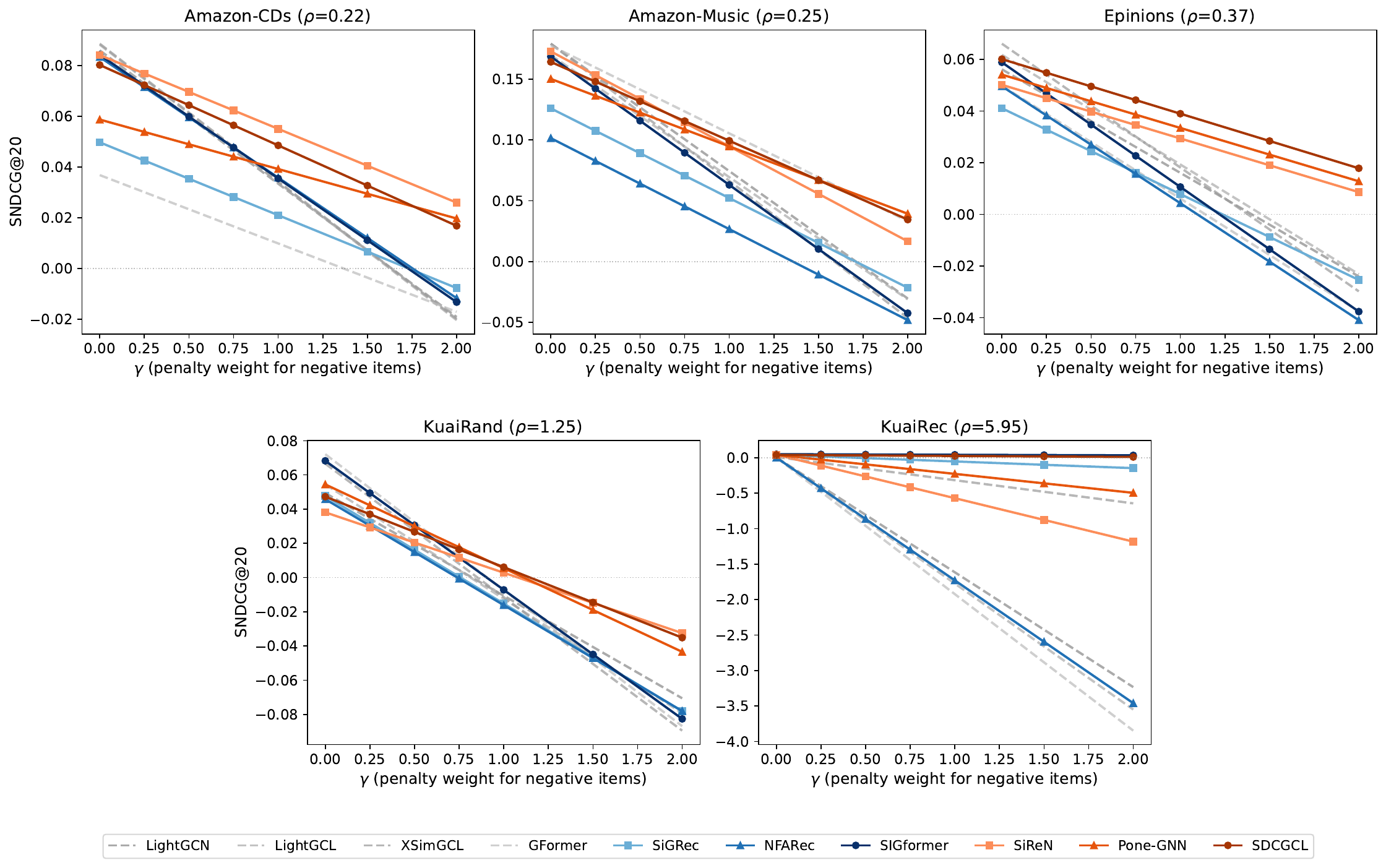}
\caption{SNDCG@20 as a function of $\gamma$ across all five benchmarks. Dashed gray: unsigned baselines; colored solid: sign-aware methods. At $\gamma\!=\!0$, SNDCG reduces to conventional NDCG (Proposition~\ref{prop:characterization}(i)); larger $\gamma$ penalizes negatives in the top-$K$ more heavily, exposing ranking changes invisible at $\gamma\!=\!0$. Note the substantially different y-axis scale on KuaiRec, reflecting the catastrophic valence failure of unsigned methods and NFARec when negatives dominate the interaction pool.}
\label{fig:gamma}
\end{figure*}

Figure~\ref{fig:gamma} shows SNDCG@20 as a function of $\gamma$ across the five benchmarks. The qualitative behavior splits cleanly along the negative-to-positive ratio $\rho$.

\textbf{Low-to-mid-$\rho$ datasets (Amazon-CDs, Amazon-Music, Epinions, KuaiRand).} At $\gamma\!=\!0$, unsigned baselines occupy the top tier across all four datasets, with SIGformer typically joining them and SiReN or SDCGCL also reaching the top tier on Music and Epinions respectively. As $\gamma$ increases, each curve's slope reflects top-$K$ negative infiltration: heavy-infiltration methods (unsigned baselines, NFARec, SIGformer) decline steeply, while SiReN, Pone-GNN, and SDCGCL decline gradually. The curves cross and the sign-aware Loss-driven methods (SiReN, Pone-GNN, SDCGCL) take the top positions, with crossover points varying by dataset, occurring around $\gamma\!\approx\!0.5$ on CDs and Epinions, $\gamma\!\approx\!0.75$ on Music, and between $\gamma\!=\!0.5$ and $\gamma\!=\!1.25$ on KuaiRand. By $\gamma\!=\!2$, the relative ordering reverses decisively: SiReN, Pone-GNN, and SDCGCL hold the top three SNDCG values on every dataset, remaining positive on CDs, Music, and Epinions, and turning only marginally negative on KuaiRand where denser negative feedback drags all methods below zero. This consistent dominance across $\gamma$ confirms that their advantage under signed metrics is robust to the specific $\gamma$ choice.

\textbf{High-$\rho$ regime (KuaiRec, $\rho\!=\!5.95$).} Methods separate dramatically across three tiers rather than via gradual crossover. With negatives outnumbering positives by nearly 6:1, placing a negative item in the top-$K$ displaces many more potentially relevant items than in low-$\rho$ regimes. The near-fully-infiltrated methods (LightGCN, LightGCL, GFormer, and NFARec, all with 99--100\% infiltration in Table~\ref{tab:neg_infiltration}) collapse to roughly $-3.5$ at $\gamma\!=\!2$, behaving as if effectively blind to valence on this dataset (Table~\ref{tab:full_diagnosis}). At the other extreme, SIGformer (with $\beta\!=\!-0.2$), SDCGCL, and SiGRec remain close to zero throughout the sweep, reflecting their substantially lower infiltration. The remaining methods occupy an intermediate tier: XSimGCL (50.8\% infiltration), Pone-GNN (69.4\%), and SiReN (83.8\%) reach $\gamma\!=\!2$ SNDCG values of approximately $-0.64$, $-0.50$, and $-1.18$ respectively. This three-way separation reflects the absence of any unsigned method that achieves meaningful valence separation on KuaiRec; at high $\rho$ the evaluation sharply distinguishes methods that engage with valence from those that do not.

\subsection{Negative Infiltration Statistics}
\label{app:gamma_infiltration}

The slope of each curve in Figure~\ref{fig:gamma} is determined by how often a method places negative items in the top-$K$. Table~\ref{tab:neg_infiltration} reports a closely related summary statistic: ``Neg $>$ Pos'' is the fraction of users with at least one negative item ranked above a positive item, which captures the user-experienced harm of negative infiltration. Across the four low-to-mid-$\rho$ datasets, sign-aware Loss-driven methods (SiReN, Pone-GNN, SDCGCL) consistently exhibit lower infiltration rates than unsigned baselines and Architecture-driven sign-aware methods, mirroring their flatter SNDCG curves under increasing $\gamma$. KuaiRec is again the exception: SIGformer (with $\beta\!=\!-0.2$), SDCGCL, and SiGRec achieve near-zero infiltration ($\leq\!8.5\%$), while LightGCN, LightGCL, GFormer, and NFARec reach 99--100\%, with XSimGCL, SiReN, and Pone-GNN occupying intermediate values (50.8\%, 83.8\%, 69.4\%). This pattern is consistent with the three-tier separation observed in Figure~\ref{fig:gamma}.

\begin{table}[t]
\caption{Negative infiltration in top-20 recommendations ($K\!=\!20$). ``Neg $>$ Pos'' is the fraction (\%) of users with at least one negative item ranked above a positive item. Results are mean\textsubscript{$\pm$std} over three seeds. Dataset headers show $\rho$ in parentheses.}
\label{tab:neg_infiltration}
\centering
\footnotesize
\setlength{\tabcolsep}{4pt}
\renewcommand{\arraystretch}{0.9}
\begin{tabular}{ll l ccccc}
\toprule
\multicolumn{2}{c}{\textbf{Group}} & \textbf{Method} & \textbf{CDs (0.22)} & \textbf{Music (0.25)} & \textbf{Epin.\ (0.37)} & \textbf{KuaiR.\ (1.25)} & \textbf{KuaiRec (5.95)} \\
\midrule
\multirow{4}{*}{Unsigned} &
  & LightGCN    & 11.0\textsubscript{$\pm$.1} & 17.0\textsubscript{$\pm$.7}  & 12.2\textsubscript{$\pm$.1} & 17.0\textsubscript{$\pm$1.3} & 99.3\textsubscript{$\pm$1.1}  \\
& & GFormer     & 7.4\textsubscript{$\pm$.1}  & 14.8\textsubscript{$\pm$1.9} & 13.1\textsubscript{$\pm$.0} & 18.8\textsubscript{$\pm$.2}  & \textbf{100.0}                \\
& & LightGCL    & 11.5\textsubscript{$\pm$.2} & 15.4\textsubscript{$\pm$1.1} & 11.5\textsubscript{$\pm$.3} & 15.4\textsubscript{$\pm$1.0} & 99.9\textsubscript{$\pm$.1}   \\
& & XSimGCL     & 11.1\textsubscript{$\pm$.3} & 16.8\textsubscript{$\pm$.7}  & 12.9\textsubscript{$\pm$.2} & 17.8\textsubscript{$\pm$.4}  & 50.8\textsubscript{$\pm$1.6}  \\
\cmidrule(lr){1-8}
\multirow{6}{*}{Signed} & \multirow{3}{*}{Arch.}
  & SiGRec      & 8.0\textsubscript{$\pm$.4}  & 10.8\textsubscript{$\pm$5.4} & 11.0\textsubscript{$\pm$.1} & 16.5\textsubscript{$\pm$.3}  & 8.5\textsubscript{$\pm$.6}    \\
& & NFARec      & 10.3\textsubscript{$\pm$.1} & 13.3\textsubscript{$\pm$1.0} & 14.2\textsubscript{$\pm$.1} & 15.7\textsubscript{$\pm$.3}  & 99.8\textsubscript{$\pm$.3}   \\
& & SIGformer   & 10.5\textsubscript{$\pm$.2} & 17.3\textsubscript{$\pm$.6}  & 13.2\textsubscript{$\pm$.6} & 17.8\textsubscript{$\pm$.2}  & 3.1\textsubscript{$\pm$.4}    \\
\cmidrule(lr){2-8}
& \multirow{3}{*}{Loss}
  & SiReN       & 6.6\textsubscript{$\pm$.1}  & 13.3\textsubscript{$\pm$.2}  & 6.5\textsubscript{$\pm$.3}  & 9.4\textsubscript{$\pm$.7}   & 83.8\textsubscript{$\pm$28.1} \\
& & Pone-GNN    & 5.6\textsubscript{$\pm$.2}  & 9.6\textsubscript{$\pm$.4}   & 6.7\textsubscript{$\pm$.2}  & 12.6\textsubscript{$\pm$.3}  & 69.4\textsubscript{$\pm$12.0} \\
& & SDCGCL      & 5.1\textsubscript{$\pm$1.8} & 11.2\textsubscript{$\pm$.2}  & 5.7\textsubscript{$\pm$.5}  & 10.8\textsubscript{$\pm$.6}  & 5.7\textsubscript{$\pm$.6}    \\
\bottomrule
\end{tabular}
\end{table}

\subsection{Choosing \texorpdfstring{$\gamma$}{gamma} in Practice}
\label{app:gamma_choice}

The penalty weight $\gamma$ operationalizes the relative cost of recommending a negative item versus the reward of recommending a positive one. A simple decision-theoretic view sets $\gamma\!=\!|u_\text{neg}|/u_\text{pos}$, the ratio of the per-exposure cost of a negative item to the per-exposure reward of a positive one. The symmetric default $\gamma\!=\!1$ treats a negative-item intrusion as canceling a positive-item recommendation of equal rank (a neutral benchmark). Asymmetric cost structures admit $\gamma\!>\!1$ (dissatisfaction-costly domains: news, advertising, safety-sensitive recommendation) or $\gamma\!<\!1$ (exploratory browsing); since $\gamma$ is an evaluation-time parameter, practitioners may report multiple values.

\section{Statistical Significance of Re-evaluation}
\label{app:sig_test}

This appendix supplements the re-evaluation results of \S\ref{sec:reeval} with detailed paired hypothesis tests on user-level scores. We verify that the conclusions drawn from Table~\ref{tab:reeval} are robust by recomputing significance on three datasets spanning low, medium, and high $\rho$ (Amazon-CDs, KuaiRand, KuaiRec). For each method pair within each column, we run a paired $t$-test on user-level scores, Holm--Bonferroni corrected within each column at family-wise $\alpha\!=\!0.05$; seed averaging precedes testing to avoid independence violations from pooling per-seed scores. Table~\ref{tab:sig_all} reports the full results, with $B/W$ counts indicating how many of the nine other methods the row method significantly \textbf{B}eats / is \textbf{W}orse than under each metric.

\textbf{Reranking is significant, not noise.} The top-SNDCG sign-aware methods all carry $\ddagger$ markers (significantly outperform every unsigned baseline): SiReN on CDs (also $\dagger$, beating all nine others), SiReN, Pone-GNN, and SDCGCL on KuaiRand, and SIGformer on KuaiRec (also $\dagger$). These markers establish that the NDCG-to-SNDCG reranking is a real shift in ranking dominance, not noise from between-user heterogeneity.

\textbf{The reordering is consistent with the significance structure.} On CDs, XSimGCL significantly outperforms 8 of 9 other methods on NDCG but only 2 of 9 on SNDCG, while SiReN moves from 4/9 on NDCG to 9/9 on SNDCG. On KuaiRand, GFormer is the unique NDCG leader (9/9, $\dagger$) yet carries no SNDCG markers, while SiReN, Pone-GNN, and SDCGCL all attain $\ddagger$ on SNDCG despite their lower NDCG rankings. The shift in significance counts directly mirrors the change in metric.

\textbf{KuaiRec amplifies the pattern.} SIGformer attains $\ddagger\dagger$ on both metrics; SiGRec, SiReN, Pone-GNN, and SDCGCL also reach $\ddagger$ on NDCG, showing that on this high-$\rho$ dataset unsigned methods with near-zero NDCG fail catastrophically rather than noisily. Because all comparisons are paired at the user level, these conclusions are robust to the substantial between-user heterogeneity that dominates the standard deviations reported in Table~\ref{tab:reeval}.

\begin{table}[t]
\caption{Statistical significance on Amazon-CDs, KuaiRand, and KuaiRec. Per-user NDCG@20 and SNDCG@20 (mean over 3 seeds). $B/W$ count how many of the nine other methods the row method significantly outperforms / is outperformed by under Holm-corrected paired $t$-test at $p\!<\!0.05$. $\ddagger$: significantly better than every unsigned baseline (sign-aware rows only). $\dagger$: significantly better than all other nine methods.}
\label{tab:sig_all}
\centering
\footnotesize
\setlength{\tabcolsep}{4pt}
\renewcommand{\arraystretch}{0.9}
\begin{tabular}{l ll l ccc ccc}
\toprule
& & & & \multicolumn{3}{c}{\textbf{NDCG@20}} & \multicolumn{3}{c}{\textbf{SNDCG@20}} \\
\cmidrule(lr){5-7} \cmidrule(lr){8-10}
\textbf{Dataset} & \multicolumn{2}{c}{\textbf{Group}} & \textbf{Method} & \textbf{Mean} & $B$ & $W$ & \textbf{Mean} & $B$ & $W$ \\
\midrule
\multirow{10}{*}{\rotatebox{90}{Amazon-CDs}}
& \multirow{4}{*}{Unsigned} &
  & LightGCN  & .0857                          & 5 & 2 & .0337                          & 2 & 2 \\
& & & GFormer   & .0370                          & 0 & 9 & .0099                          & 0 & 9 \\
& & & LightGCL  & .0884                          & 7 & 0 & .0339                          & 2 & 2 \\
& & & XSimGCL   & .0886                          & 8 & 0 & .0347                          & 2 & 2 \\
\cmidrule(lr){2-10}
& \multirow{6}{*}{Signed} & \multirow{3}{*}{Arch.}
  & SiGRec    & .0506                          & 1 & 8 & .0203                          & 1 & 8 \\
& & & NFARec    & .0840                          & 5 & 2 & .0363                          & 2 & 2 \\
& & & SIGformer & .0865                          & 5 & 0 & .0360                          & 2 & 2 \\
\cmidrule(lr){3-10}
& & \multirow{3}{*}{Loss}
  & SiReN     & .0810                          & 4 & 5 & .0526$^{\ddagger\dagger}$      & 9 & 0 \\
& & & Pone-GNN  & .0641                          & 3 & 6 & .0429$^\ddagger$               & 8 & 1 \\
& & & SDCGCL    & .0550                          & 2 & 7 & .0341                          & 2 & 2 \\
\midrule
\multirow{10}{*}{\rotatebox{90}{KuaiRand}}
& \multirow{4}{*}{Unsigned} &
  & LightGCN  & .0586                          & 6 & 3 & $-$.0101                       & 2 & 3 \\
& & & GFormer   & .0718$^\dagger$                & 9 & 0 & $-$.0082                       & 2 & 3 \\
& & & LightGCL  & .0507                          & 3 & 4 & $-$.0112                       & 0 & 3 \\
& & & XSimGCL   & .0662                          & 7 & 1 & $-$.0115                       & 0 & 3 \\
\cmidrule(lr){2-10}
& \multirow{6}{*}{Signed} & \multirow{3}{*}{Arch.}
  & SiGRec    & .0495                          & 3 & 4 & $-$.0143                       & 0 & 6 \\
& & & NFARec    & .0463                          & 2 & 7 & $-$.0157                       & 0 & 7 \\
& & & SIGformer & .0680                          & 7 & 1 & $-$.0075                       & 3 & 3 \\
\cmidrule(lr){3-10}
& & \multirow{3}{*}{Loss}
  & SiReN     & .0374                          & 0 & 9 & .0039$^\ddagger$               & 7 & 0 \\
& & & Pone-GNN  & .0516                          & 3 & 4 & .0045$^\ddagger$               & 7 & 0 \\
& & & SDCGCL    & .0428                          & 1 & 8 & .0039$^\ddagger$               & 7 & 0 \\
\midrule
\multirow{10}{*}{\rotatebox{90}{KuaiRec}}
& \multirow{4}{*}{Unsigned} &
  & LightGCN  & .0015                          & 2 & 6 & $-$1.828                       & 1 & 8 \\
& & & GFormer   & .0007                          & 0 & 7 & $-$1.931                       & 0 & 9 \\
& & & LightGCL  & .0009                          & 0 & 6 & $-$1.781                       & 2 & 7 \\
& & & XSimGCL   & .0083                          & 4 & 5 & $-$.365                        & 6 & 3 \\
\cmidrule(lr){2-10}
& \multirow{6}{*}{Signed} & \multirow{3}{*}{Arch.}
  & SiGRec    & .0412$^\ddagger$               & 7 & 1 & $-$.057$^\ddagger$             & 7 & 2 \\
& & & NFARec    & .0007                          & 0 & 7 & $-$1.622                       & 3 & 6 \\
& & & SIGformer & .0496$^{\ddagger\dagger}$      & 9 & 0 & .0418$^{\ddagger\dagger}$      & 9 & 0 \\
\cmidrule(lr){3-10}
& & \multirow{3}{*}{Loss}
  & SiReN     & .0153$^\ddagger$               & 5 & 4 & $-$1.345                       & 4 & 5 \\
& & & Pone-GNN  & .0314$^\ddagger$               & 6 & 3 & $-$.652                        & 5 & 4 \\
& & & SDCGCL    & .0382$^\ddagger$               & 6 & 1 & .0243$^\ddagger$               & 8 & 1 \\
\bottomrule
\end{tabular}
\end{table}

\section{Evaluation Subset Validation}
\label{app:subset_validation}

All re-evaluation results in \S\ref{sec:reeval} (and the $\mathcal{L}_\text{sign}$ analysis in the proof-of-concept experiment, \S\ref{sec:exp}) are computed on the ``Both P\&N'' subset: users with at least one positive and one negative test item. This restriction reflects the spirit of backward compatibility (Proposition~\ref{prop:characterization}(i)): when $\mathcal{N}_u\!=\!\emptyset$, signed metrics coincide with their conventional counterparts at every $\gamma$, so users without negative feedback contribute no signed-vs-conventional distinction. Conventional metrics are computed on the same subpopulation for fairness. A natural concern is whether this subset restriction alters the method ranking that would be obtained on the full user population. We test this directly by recomputing conventional NDCG@20 on the full test set and comparing the induced method ranking (over all ten methods) against the ranking on the Both P\&N subset.

\begin{table}[t]
\caption{Validation of the Both P\&N evaluation subset. ``Both/Full'' is the fraction of test users retained by the subset restriction. $\rho_s$ is the Spearman rank correlation and $\tau$ is the Kendall rank correlation between method rankings induced by NDCG@20 on the subset vs.\ on the full test set, computed over all ten methods (4 unsigned + 6 sign-aware). $\Delta$NDCG is the average difference (subset $-$ full-set) across methods.}
\label{tab:subset_validation}
\centering
\small
\begin{tabular}{l cccc}
\toprule
\textbf{Dataset} ($\rho$) & \textbf{Both/Full} & \textbf{Spearman $\rho_s$} & \textbf{Kendall $\tau$} & $\Delta$\textbf{NDCG} \\
\midrule
Amazon-CDs (0.22)    & 31.2\%  & 0.903\textsuperscript{***} & 0.822 & $+$.003 \\
Amazon-Music (0.25)  & 31.6\%  & 0.988 & 0.956 & $+$.003 \\
Epinions (0.37)      & 55.5\%  & 0.988 & 0.956 & $+$.003 \\
KuaiRand (1.25)      & 71.1\%  & 1.000 & 1.000 & $-$.002 \\
KuaiRec (5.95)       & 100.0\% & 1.000 & 1.000 & n/a \\
\midrule
\textbf{Mean}        &         & \textbf{0.976} & \textbf{0.947} & \\
\bottomrule
\multicolumn{5}{l}{\footnotesize \textsuperscript{***} $p < 10^{-3}$; other correlations are $p < 10^{-4}$ (exact test, $n\!=\!10$).} \\
\multicolumn{5}{p{0.9\columnwidth}}{\footnotesize The ``Both/Full'' fractions here are measured at evaluation time and use as denominator only test users with at least one ranking-eligible item (any sign) in the candidate pool, whereas the ``Both P\&N'' column in Table~\ref{tab:datasets} uses all raw test users as denominator. The two differ when some test users have no ranking-eligible items under the evaluation pipeline (5-core filtering plus candidate-pool constraints) and are therefore excluded from the denominator here; on KuaiRec the Both/Full reaches 100\% because every evaluation-eligible test user happens to have items of both signs.}
\end{tabular}
\end{table}

The table shows that the subset restriction preserves method ranking with mean Spearman $\rho_s\!=\!0.976$ across datasets. The minimum correlation is on Amazon-CDs ($\rho_s\!=\!0.903$, still significant at $p < 10^{-3}$), where three mid-tier sign-aware methods (SiGRec, SDCGCL, Pone-GNN) shuffle positions between the two rankings; the top-ranked method (XSimGCL) and the bottom-ranked method (GFormer) remain stable on both rankings. On all other datasets, $\rho_s \geq 0.988$, with KuaiRand and KuaiRec reaching the maximum 1.000 (on KuaiRec, every test user qualifies for the Both P\&N subset, so the two rankings are definitionally identical).

The $\Delta$NDCG column reports the direction of the subset's effect on absolute scores. On the three low-$\rho$ datasets (CDs, Music, Epinions), Both P\&N users score approximately $+0.003$ NDCG higher on average than the full population: users who express both positive and negative feedback are slightly more active and their preferences are more clearly signaled. On KuaiRand ($\rho\!=\!1.25$), the direction reverses ($-0.002$): with denser negative feedback, the Both P\&N users have larger candidate pools of competing items and recommendation becomes marginally harder. These shifts are small relative to the method-to-method differences that drive ranking conclusions (typical NDCG gap of 0.01--0.05 across methods), and they are consistent in sign with the underlying dataset structure rather than indicative of selection bias against any particular method class.

We therefore conclude that our evaluation protocol satisfies the backward-compatibility constraint motivated by Proposition~\ref{prop:characterization}(i) without introducing a distortion of method ranking. Readers interested in absolute conventional-metric performance on the full population should interpret the numbers in Table~\ref{tab:reeval} with the small, dataset-specific offsets above in mind; the relative conclusions (which are what the paper's arguments rest on) are unaffected.

\section{Gradient Analysis of Valence Injection}
\label{app:gradient}

This appendix formalizes why the auxiliary loss $\mathcal{L}_\text{sign}$ introduced in the proof-of-concept experiment (\S\ref{sec:exp}) succeeds where natural alternatives fail. We analyze the gradient structure of two loss designs (\S\ref{app:gradient_setup}--\S\ref{app:gradient_decoupling}), provide geometric intuition (\S\ref{app:gradient_geometry}), document additional valence-injection approaches that we explored but discarded (\S\ref{app:gradient_alternatives}), and empirically validate the two-stage (training + inference) view via ablation (\S\ref{app:gradient_ablation}).

\subsection{Setup}
\label{app:gradient_setup}

For a negative item $n$ with embedding $\mathbf{e}_n$, any auxiliary valence loss $\mathcal{L}_\text{val}$ updates $\mathbf{e}_n$ by $-\eta \, \partial \mathcal{L}_\text{val} / \partial \mathbf{e}_n$, altering the relevance score $s_{un} = \mathbf{e}_u^\top \mathbf{e}_n$ by
\begin{equation}
\label{eq:score_change}
\Delta s_{un} = -\eta \, \mathbf{e}_u^\top \frac{\partial \mathcal{L}_\text{val}}{\partial \mathbf{e}_n}.
\end{equation}
If this projection is non-zero, the valence loss inevitably alters relevance. We compare two loss designs along this criterion.

\subsection{Full Relevance Coupling of Direct Separation}
\label{app:gradient_coupling}

\begin{proposition}[Full Relevance Coupling of Direct Separation]
\label{prop:pn_coupling}
For the direct separation loss $\mathcal{L}_\text{P>N} = -\ln \sigma(\hat{y}_{up} - \hat{y}_{un})$, the gradient with respect to $\mathbf{e}_n$ is proportional to $\mathbf{e}_u$. Therefore, by Eq.~\ref{eq:score_change}, every non-trivial gradient step necessarily modifies the relevance score $s_{un}$.
\end{proposition}

\begin{proof}[Proof of Proposition~\ref{prop:pn_coupling}]
Since $\hat{y}_{un} = \mathbf{e}_u^\top \mathbf{e}_n$, we have $\partial \hat{y}_{un} / \partial \mathbf{e}_n = \mathbf{e}_u$. Applying the chain rule:
\begin{equation}
\frac{\partial \mathcal{L}_\text{P>N}}{\partial \mathbf{e}_n} = -\frac{\partial}{\partial \hat{y}_{un}} \ln \sigma(\hat{y}_{up} - \hat{y}_{un}) \cdot \frac{\partial \hat{y}_{un}}{\partial \mathbf{e}_n} = \big(1 - \sigma(\hat{y}_{up} - \hat{y}_{un})\big) \, \mathbf{e}_u
\end{equation}
The coefficient $(1 - \sigma(\cdot)) \in (0, 1)$ is strictly positive, so the gradient is a positive scalar multiple of $\mathbf{e}_u$. Substituting into Eq.~\ref{eq:score_change} yields $\Delta s_{un} = -\eta (1 - \sigma(\cdot)) \|\mathbf{e}_u\|^2 \neq 0$.
\end{proof}

\subsection{Partial Relevance Decoupling of Sign Classification}
\label{app:gradient_decoupling}

\begin{proposition}[Partial Relevance Decoupling of Sign Classification]
\label{prop:sign_decoupling}
Consider the negative-interaction case ($y = 0$) of the sign classification loss, which reduces to $\mathcal{L}_\text{sign} = -\ln(1 - \hat{p}_{un})$ with $\hat{p}_{un} = \sigma(\mathbf{w}^\top (\mathbf{e}_u \odot \mathbf{e}_n) + b)$. The gradient with respect to $\mathbf{e}_n$ is proportional to $\mathbf{w} \odot \mathbf{e}_u$, which has a non-zero component orthogonal to $\mathbf{e}_u$ whenever $\mathbf{w}$ is not a constant vector. This orthogonal component modifies the element-wise pattern of $\mathbf{e}_n$ without altering $\mathbf{e}_u^\top \mathbf{e}_n$. The positive-interaction case ($y=1$) is analogous with the sign of the gradient reversed.
\end{proposition}

\begin{proof}[Proof of Proposition~\ref{prop:sign_decoupling}]
Since $\hat{p}_{un} = \sigma(\mathbf{w}^\top (\mathbf{e}_u \odot \mathbf{e}_n) + b)$ and $\partial(\mathbf{e}_u \odot \mathbf{e}_n)/\partial\mathbf{e}_n = \text{diag}(\mathbf{e}_u)$, by the chain rule:
\begin{equation}
\frac{\partial \mathcal{L}_\text{sign}}{\partial \mathbf{e}_n} = \hat{p}_{un} \, \text{diag}(\mathbf{e}_u) \, \mathbf{w} = \hat{p}_{un} \, (\mathbf{w} \odot \mathbf{e}_u)
\end{equation}
Decomposing $\mathbf{w} \odot \mathbf{e}_u$ into components parallel and orthogonal to $\mathbf{e}_u$:
\begin{equation}
\mathbf{w} \odot \mathbf{e}_u = \underbrace{\frac{\sum_k w_k e_{uk}^2}{\|\mathbf{e}_u\|^2} \mathbf{e}_u}_{\text{parallel (changes } s_{un}\text{)}} + \underbrace{\left(\mathbf{w} \odot \mathbf{e}_u - \frac{\sum_k w_k e_{uk}^2}{\|\mathbf{e}_u\|^2} \mathbf{e}_u\right)}_{\text{orthogonal (preserves } s_{un}\text{)}}
\end{equation}
The orthogonal component is zero if and only if $\mathbf{w} \odot \mathbf{e}_u \propto \mathbf{e}_u$, which holds only when all entries of $\mathbf{w}$ are equal. Since $\mathbf{w}$ is learned and generically non-constant, the orthogonal component is non-zero.
\end{proof}

The key distinction: $\mathcal{L}_\text{P>N}$ operates entirely in the direction of $\mathbf{e}_u$, so every step improving valence separation also changes the relevance score, a gradient-conflict pattern familiar from multi-task optimization~\cite{yu2020pcgrad}. $\mathcal{L}_\text{sign}$ operates in the direction of $\mathbf{w} \odot \mathbf{e}_u$, generically with an orthogonal component that carries valence information without altering relevance. $\mathcal{L}_\text{sign}$ does not \emph{guarantee} relevance preservation (its parallel component still modifies $s_{un}$), but it provides a degree of freedom that $\mathcal{L}_\text{P>N}$ structurally lacks.

\subsection{Geometric Intuition}
\label{app:gradient_geometry}

The inner product $s_{un} = \mathbf{e}_u^\top \mathbf{e}_n = \sum_k e_{uk} e_{nk}$ is a single scalar that sums over all $d$ embedding dimensions. Any perturbation $\Delta$ to $\mathbf{e}_n$ that is orthogonal to $\mathbf{e}_u$ (i.e., $\mathbf{e}_u^\top \Delta = 0$) leaves this sum unchanged while modifying the individual dimension values. The element-wise product $\mathbf{e}_u \odot \mathbf{e}_n$ retains all $d$ dimension-level values, so a linear classifier $\mathbf{w}^\top (\mathbf{e}_u \odot \mathbf{e}_n)$ can detect changes in individual dimensions that are invisible to the inner product.

Concretely, $\mathcal{L}_\text{sign}$ can be satisfied by adjusting the \emph{relative magnitudes} across embedding dimensions without changing the overall inner product. For example, if positive interactions activate certain dimensions more strongly and negative interactions activate others, the linear classifier can detect this pattern while the inner product (a sum over all dimensions) remains largely unaffected.

$\mathcal{L}_\text{P>N}$, by contrast, produces gradients entirely in the $\mathbf{e}_u$ direction (Proposition~\ref{prop:pn_coupling}), which has no orthogonal component. Every gradient step that separates valence simultaneously modifies the relevance score.

\subsection{Alternative Valence Injection Approaches}
\label{app:gradient_alternatives}

In developing $\mathcal{L}_\text{sign}$, we explored several alternatives, each illustrating a different manifestation of the coupling identified in Proposition~\ref{prop:pn_coupling}. For reference, the signed BPR objective used by SIGformer (our $\mathcal{L}_\text{base}$) takes the form
\begin{equation}
\label{eq:signed_bpr}
\mathcal{L}_\text{base} = -\!\!\!\!\!\sum_{\substack{(u,p) \in \mathcal{E}^+ \\ j \sim \mathcal{U}_u}}\!\!\!\!\! \underbrace{\ln \sigma\!\big(\hat{y}_{up} - \hat{y}_{uj}\big)}_{\mathcal{L}_\text{pos}} \;-\!\!\!\!\!\sum_{\substack{(u,n) \in \mathcal{E}^- \\ j \sim \mathcal{U}_u}}\!\!\!\!\! \underbrace{\ln \sigma\!\big(\beta \cdot (\hat{y}_{un} - \hat{y}_{uj})\big)}_{\mathcal{L}_\text{neg}},
\end{equation}
where $p, n, j$ denote a positive, negative, and sampled unobserved item, and $\beta \in [-1, 1]$ controls the direction of the negative-feedback signal. We refer to the second sum (the negative-feedback term) as $\mathcal{L}_\text{neg}$ in what follows.

\textbf{Direct score separation.} Adding a BPR term $\mathcal{L}_\text{P>N} = -\sum \ln \sigma(\hat{y}_{up} - \hat{y}_{un})$ conflicts with the base model's negative term $\mathcal{L}_\text{neg}$ (second sum in Eq.~\ref{eq:signed_bpr}). The two push $\hat{y}_{un}$ in opposite directions on the same parameters at every step (assuming $\beta > 0$, which holds on four of the five benchmarks):
\begin{align}
\frac{\partial \mathcal{L}_\text{neg}}{\partial \hat{y}_{un}} &< 0 \quad \text{(pushes $\hat{y}_{un}$ up)}, \\
\frac{\partial \mathcal{L}_\text{P>N}}{\partial \hat{y}_{un}} &> 0 \quad \text{(pushes $\hat{y}_{un}$ down)}.
\end{align}

\textbf{Score space decomposition.} Splitting the embedding into relevance and valence subspaces ($\mathbf{e} = [\mathbf{e}^\text{rel} | \mathbf{e}^\text{val}]$) and computing $\hat{y} = \mathbf{e}_u^\text{rel} \cdot \mathbf{e}_i^\text{rel} + \mathbf{e}_u^\text{val} \cdot \mathbf{e}_i^\text{val}$ did not resolve the conflict, as both subspaces share gradients through the combined score.

\textbf{Dual-encoder architectures.} Using separate encoders for relevance (e.g., LightGCN on unsigned graph) and valence (e.g., SIGformer on signed graph) avoids gradient mixing, but the relevance encoder proved weaker than the original SIGformer, resulting in net performance loss.

\textbf{Gradient detachment.} Computing valence loss on detached embeddings ($\mathbf{e}.\text{detach()}$) prevents valence gradients from affecting relevance learning, but also prevents valence from influencing the model at all, rendering the auxiliary task ineffective.

The common failure mode across all approaches is that their gradients are either fully parallel to $\mathbf{e}_u$ (as in direct separation, Proposition~\ref{prop:pn_coupling}) or mix relevance and valence directions in ways that cannot be controlled. $\mathcal{L}_\text{sign}$ addresses this by introducing a gradient component orthogonal to $\mathbf{e}_u$ (Proposition~\ref{prop:sign_decoupling}), providing a channel for valence injection that is structurally absent in the alternatives.

\subsection{Empirical Validation: \texorpdfstring{$\mathcal{L}_\text{sign}$}{L\_sign} Ablation}
\label{app:gradient_ablation}

We empirically validate the two-stage view implied by Propositions~\ref{prop:pn_coupling}--\ref{prop:sign_decoupling}: $\mathcal{L}_\text{sign}$ enriches the embedding space at training time (via the orthogonal gradient component), and an inference-time scoring adjustment $\lambda_\text{inf}$ extracts the valence information that the inner product alone cannot read (Appendix~\ref{app:impl_lsign} for the inference-time formula). We disentangle the two mechanisms by varying each hyperparameter while fixing the other, on KuaiRand because it exhibits both substantial negative infiltration (Table~\ref{tab:neg_infiltration}) and non-trivial baseline SNDCG.

\begin{figure}[h]
    \centering
    \begin{subfigure}[t]{0.48\linewidth}
        \centering
        \includegraphics[width=\linewidth]{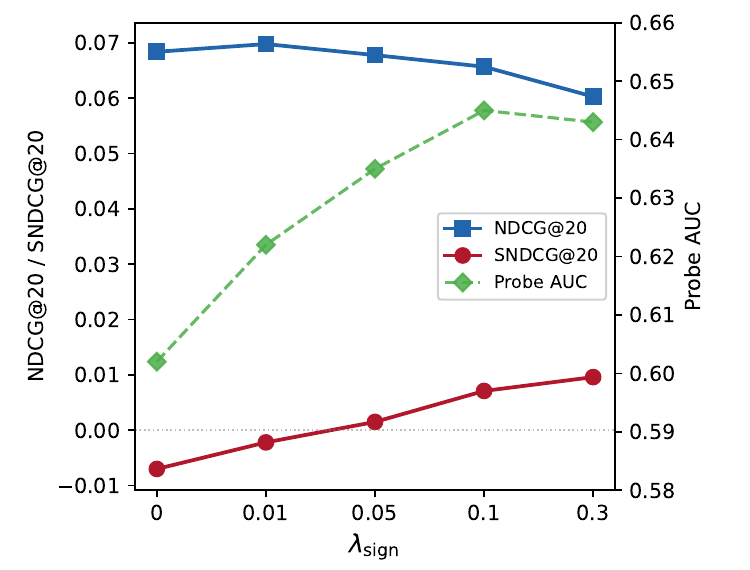}
        \caption{Training effect ($\lambda_\text{inf}\!=\!2$ fixed).}
        \label{fig:ablation_a}
    \end{subfigure}
    \hfill
    \begin{subfigure}[t]{0.48\linewidth}
        \centering
        \includegraphics[width=\linewidth]{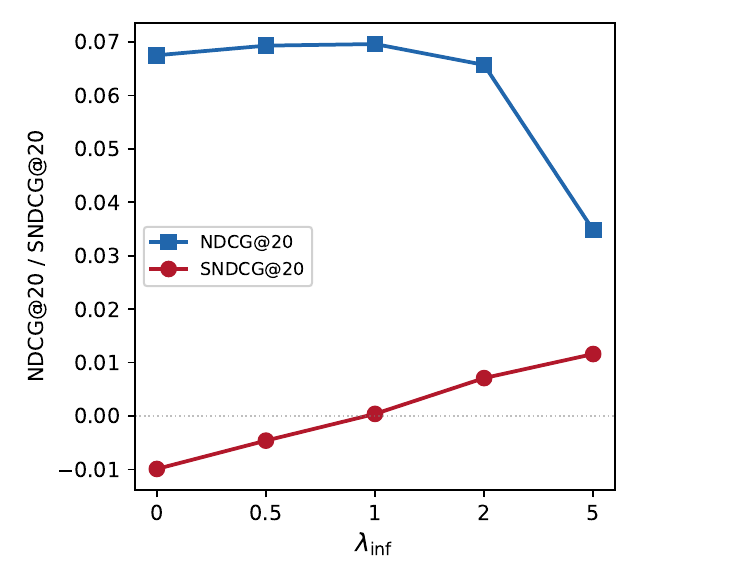}
        \caption{Inference effect ($\lambda_\text{sign}\!=\!0.1$ fixed).}
        \label{fig:ablation_b}
    \end{subfigure}
    \caption{Ablation of $\mathcal{L}_\text{sign}$ on KuaiRand ($K\!=\!20$), Both P\&N pool. Dashed gray line: SNDCG~$=0$.}
    \label{fig:ablation}
\end{figure}

\textbf{Training effect (Figure~\ref{fig:ablation_a}).} As $\lambda_\text{sign}$ increases with $\lambda_\text{inf}\!=\!2$, probe AUC rises from $0.602$ to $0.645$ and SNDCG improves correspondingly, crossing zero between $\lambda_\text{sign}\!=\!0.01$ and $0.05$. NDCG drops modestly (from $0.0684$ to $0.0603$ at $\lambda_\text{sign}\!=\!0.3$), reflecting the parallel gradient component identified in Proposition~\ref{prop:sign_decoupling}: even though $\mathcal{L}_\text{sign}$ has an orthogonal channel, its parallel component still perturbs the inner-product score. Probe AUC plateaus at $\lambda_\text{sign}\!\geq\!0.1$, indicating that valence information saturates the embedding space.

\textbf{Inference effect (Figure~\ref{fig:ablation_b}).} With $\lambda_\text{sign}\!=\!0.1$ fixed, probe AUC stays constant at $0.645$ (the embeddings are unchanged), while SNDCG rises monotonically from $-0.0099$ to $+0.0116$ as $\lambda_\text{inf}$ increases. NDCG degrades sharply at high $\lambda_\text{inf}$ (a $49\%$ drop at $\lambda_\text{inf}\!=\!5$), so $\lambda_\text{inf}\!=\!2$ is the best operating point under the $5\%$ NDCG-preservation constraint.

\textbf{Both mechanisms are independently necessary.} With $\lambda_\text{sign}\!=\!0$ there is no valence information in the embeddings to leverage, regardless of $\lambda_\text{inf}$. With $\lambda_\text{inf}\!=\!0$ the enriched embedding does not translate into ranking improvements because inner-product scoring cannot access the valence patterns, consistent with the probe-vs-V-AUC gap reported in Table~\ref{tab:probe} and the geometric intuition of \S\ref{app:gradient_geometry}. The ablation confirms the two-stage view: training injects valence into embeddings (an embedding-level operation), and inference extracts it at scoring time (a scoring-level operation).

\section{Implementation Details}
\label{app:impl}

\subsection{Common Configuration}
\label{app:impl_common}

All models share the following settings unless stated otherwise. Embedding dimension $d\!=\!64$, evaluation at $K\!=\!20$, signed metric weight $\gamma\!=\!1.0$. Optimizer is Adam ($\beta_1\!=\!0.9$, $\beta_2\!=\!0.999$). All results are averaged over three runs with seeds $\{42, 2024, 2025\}$. Datasets are split 7:1:2 (train/validation/test) per run with 5-core filtering. Interactions are labeled positive when rating $\geq 4$ (Amazon, Epinions), \texttt{is\_click}$=1$ (KuaiRand), or \texttt{watch\_ratio} $\geq 1.0$ (KuaiRec); otherwise negative. Unsigned baselines treat all observed interactions as implicit positives. Early stopping monitors validation NDCG@20 by default (exception: NFARec monitors validation loss). Table~\ref{tab:hyperparams} reports the per-model hyperparameters used in our experiments, exactly as configured in our released code and scripts; running these scripts reproduces our reported numbers. Per-method values follow each baseline's released code (the configurations the original authors provide as ready-to-run defaults), with the shared infrastructure choices above standardized across methods for fair comparison.

\begin{table*}[t]
\caption{Per-model hyperparameters. \#Params reported on Amazon-CDs ($51{,}267$ users $\times$ $46{,}464$ items, $d\!=\!64$); model size scales proportionally across datasets. Type: U = unsigned, S = sign-aware. Architecture: G = GCN, T = Graph Transformer, M = MLP, L = LightGCN conv, sGCN = signed GCN. Neg/Pos: number of BPR-style negative samples per positive; ``--'' indicates the model does not use BPR-style negative sampling. Init: N($0, \sigma$) = normal, XU = Xavier uniform, XN = Xavier normal.}
\label{tab:hyperparams}
\centering
\resizebox{0.95\textwidth}{!}{%
\begin{tabular}{l c r c l r r r r l l}
\toprule
\textbf{Model} & \textbf{Type} & \textbf{\#Params} & \textbf{LR} & \textbf{Layers} & \textbf{Epochs} & \textbf{Patience} & \textbf{Batch} & \textbf{Neg/Pos} & \textbf{Init} & \textbf{Other} \\
\midrule
LightGCN   & U & 6.3M   & 1e-3 & 3 G             & 1000 & 50  & 2048 & 1  & N(0,0.1) & decay=1e-4 \\
LightGCL   & U & 6.3M   & 1e-3 & 2 G             & 100  & 20  & 4096 & 1  & XU       & cl\_$\lambda$=0.2, temp=0.2 \\
XSimGCL    & U & 6.3M   & 1e-3 & 2 G             & 100  & 20  & 2048 & 1  & XU       & cl=0.2, eps=0.2 \\
GFormer    & U & 6.3M   & 1e-3 & 2G+1T$^\dagger$ & 100  & 20  & 4096 & 1  & XU       & ssl\_reg=1.0 \\
\midrule
SIGformer  & S & 6.3M   & 1e-2 & 3 sGCN          & 1000 & 200 & full$^\S$ & 1  & N(0,0.1) & $\alpha, \beta$ per-dataset$^\ddagger$ \\
SiReN      & S & 12.5M  & 5e-3 & 4G+2M           & 200  & 40  & 1024 & 40 & XN       & reg=0.05 \\
SiGRec     & S & 12.5M  & 1e-3 & 3 sGCN          & 1000 & 50  & 2048 & 1  & XU       & $\mu$=0.3, $\omega$=300 \\
Pone-GNN   & S & 12.5M  & 5e-3 & 4G+2M           & 1000 & 10  & 2048 & 40 & XN       & reg=5e-5, dual emb \\
SDCGCL     & S & 6.3M   & 3e-3 & 2 dual          & 100  & 20  & 2048 & 1  & XU       & eps=0.2, $\tau$=0.15 \\
NFARec     & S & 193.6M & 1e-2 & 1 THP           & 20$^\P$ & 10  & 32$^\P$ & -- & XU       & d\_model=1024, n\_head=1, emb=N(0,1) \\
\bottomrule
\multicolumn{11}{p{0.92\textwidth}}{\footnotesize $^\dagger$ GFormer on Amazon-CDs uses 3G+2T (see \S\ref{app:impl_overrides}). $^\ddagger$ SIGformer uses dataset-specific $\alpha, \beta$ (see \S\ref{app:impl_overrides}). $^\S$ SIGformer uses full-batch training; the value 1024 in the original code refers to evaluation batch only. $^\P$ NFARec uses 30 epochs on Music (else 20) and batch 16 on Music and KuaiRec (else 32).}
\end{tabular}%
}
\end{table*}

\subsection{Dataset-Specific Overrides}
\label{app:impl_overrides}

Two models use dataset-specific configurations:
\begin{itemize}
\item \textbf{SIGformer}: The signed-BPR parameters $(\alpha, \beta)$ are tuned per dataset following the original paper:
\begin{center}
\small
\begin{tabular}{lccccc}
\toprule
 & CDs & Music & Epin. & KuaiRand & KuaiRec \\
\midrule
$\alpha$ & 0.4 & 0.0 & 0.4 & 0.2 & $-0.8$ \\
$\beta$  & 1.0 & 1.0 & 1.0 & 1.0 & $-0.2$ \\
\bottomrule
\end{tabular}
\end{center}
On KuaiRec, $\beta\!=\!-0.2$ reverses the direction of negative feedback in the signed BPR loss, reflecting the original paper's recommendation for datasets with dense negative interactions.
\item \textbf{GFormer on Amazon-CDs}: 3 GCN + 2 Transformer layers (default 2 GCN + 1 Transformer). The default configuration did not converge on the larger CDs graph; a small grid search over layer counts, learning rate, and ssl\_reg yielded the 3G+2T configuration.
\end{itemize}

\subsection{Sign Classification Loss (\texorpdfstring{$\mathcal{L}_\text{sign}$}{L\_sign}) Configuration}
\label{app:impl_lsign}

$\mathcal{L}_\text{sign}$ is applied to SIGformer as the base model (proof-of-concept experiment in \S\ref{sec:exp}). The auxiliary classifier $\hat{p}_{ui} = \sigma(\mathbf{w}^\top (\mathbf{e}_u \odot \mathbf{e}_i) + b)$ adds $d + 1 = 65$ parameters on top of the base model.

\textbf{Inference-time scoring.} At inference, the classifier output $\hat{p}_{ui}$ is incorporated into the final ranking score via a valence-aware adjustment:
\begin{equation}
\label{eq:final_score}
\hat{y}_{ui}^\text{final} = \hat{y}_{ui} + \lambda_\text{inf} \cdot \big(2\hat{p}_{ui} - 1\big),
\end{equation}
where $\hat{y}_{ui} = \mathbf{e}_u^\top \mathbf{e}_i$ is the original relevance score and $(2\hat{p}_{ui} - 1) \in [-1, +1]$ acts as a valence adjustment that pulls predicted-positive items up and predicted-negative items down. Setting $\lambda_\text{inf}\!=\!0$ recovers the original SIGformer scoring, isolating the effect of embedding enrichment from the score-level adjustment; this isolation is what makes the ablation in \S\ref{app:gradient_ablation} interpretable.

\textbf{Hyperparameter search.} The grid is $\lambda_\text{sign} \in \{0.01, 0.05, 0.1, 0.3\}$ and $\lambda_\text{inf} \in \{0, 0.1, 0.5, 1, 2, 5\}$. Per dataset we select the configuration that maximizes SNDCG@20 subject to NDCG@20 $\geq 0.95 \times$ the baseline's NDCG@20; the selected $(\lambda_\text{sign}, \lambda_\text{inf})$ per dataset are reported in Table~\ref{tab:lsign}.

\subsection{Training Cost}
\label{app:impl_compute}

Per-dataset, per-seed wall-clock times on a single 24GB GPU fall into three groups. Fast (under 30 minutes): LightGCL, XSimGCL, SDCGCL, Pone-GNN, NFARec (20 epochs). Medium (1--3 hours): LightGCN, SiGRec. Slow ($\geq$2 hours): SIGformer ($\sim\!2$h), SiReN ($\sim\!3$h), GFormer ($\sim\!6$h). Total experimental compute for the results reported in this paper is approximately 140 GPU-hours across 10 methods $\times$ 5 datasets $\times$ 3 seeds, primarily on 24GB-class research GPUs. NFARec on Amazon-CDs is the exception: its default $d_\text{model}=1024$ and the resulting $46{,}464^2$ correlation matrix require $\geq 80$GB GPU memory, which we honored per the original paper. Replicators with smaller hardware may substitute $d_\text{model}\!=\!512$, which degrades NDCG on CDs by approximately 8\% in our preliminary tests but preserves the valence failure phenomenon identified in \S\ref{sec:analysis}.

\subsection{On Baseline Fairness}
\label{app:impl_fairness}

A natural concern is whether the per-method hyperparameters reported in Table~\ref{tab:hyperparams} confer differential advantage. We use each method's released code/scripts as the reproducibility-friendly configuration: original-paper hyperparameters often presuppose hardware or dataset conditions different from ours, and the released scripts encode the authors' refined defaults for general use. The values in Table~\ref{tab:hyperparams} are therefore exactly what our experiments executed, and rerunning our scripts reproduces the reported numbers. SIGformer's per-dataset $(\alpha, \beta)$ tuning (the only method that prescribes this at the level documented here) and GFormer's Amazon-CDs layer override (where the default did not converge) are documented separately in \S\ref{app:impl_overrides}. Other sign-aware methods (SiReN, SiGRec, Pone-GNN, SDCGCL, NFARec) use their respective released-code configurations without per-dataset tuning beyond what is specified in their public scripts.